%% file: main.tex
\PassOptionsToPackage{usenames,dvipsnames,svgnames,table}{xcolor}
\pdfoutput=1

\newif\iflong
\longfalse
\longtrue % uncomment to compile long version

\iflong
\documentclass[sigconf,nonacm,natbib]{acmart}
\else
\documentclass[sigconf]{acmart}
\fi

\AtBeginDocument{%
  }

\iflong
\setcopyright{cc}
\copyrightyear{2025}
\acmYear{2025}
\acmDOI{XXXXXXX.XXXXXXX}
\else
\acmYear{2025}
\copyrightyear{2025}
\acmConference[Middleware '25]{26th ACM Middleware Conference}{December 15--19, 2025}{Nashville, TN, USA}
\acmBooktitle{26th ACM Middleware Conference (Middleware '25), December 15--19, 2025, Nashville, TN, USA}
\acmDOI{10.1145/3721462.3730953}
\acmISBN{979-8-4007-1554-9/25/12}
\fi

\usepackage{format}
\usepackage{macros}
\usepackage{main}

\tikzset{near start abs/.style={xshift=10cm}}

\begin{document}

\iflong
\title[]{Adjusted Objects: An Efficient and Principled Approach\\ to Scalable Programming (Extended Version)}
\settopmatter{printacmref=false, printccs=true, printfolios=true}
\else
\title{Adjusted Objects: An Efficient and Principled Approach to Scalable Programming}
\acmConference[Middleware'25]{26th ACM/IFIP International Middleware Conference}{December 15--19}{Nashville, TN, USA}
\fi

\author{Boubacar Kane}
\orcid{0009-0009-4314-9386}
\email{boubacar.kane@telecom-sudparis.eu}
\affiliation{%
  \institution{Telecom SudParis, Inria Saclay}
  \institution{Institut Polytechnique de Paris}
  \city{Palaiseau}
  \country{France}
}

\author{Pierre Sutra}
\orcid{0000-0002-0573-2572}
\email{pierre.sutra@telecom-sudparis.eu}
\affiliation{%
  \institution{Telecom SudParis, Inria Saclay}
  \institution{Institut Polytechnique de Paris}
  \city{Palaiseau}    
  \country{France}
}

\input{abstract}





\maketitle

\input{introduction}

\input{model}
\input{distinguishability}
\input{adjusted}
\input{design}

\input{evaluation}
\input{related}

\input{conclusion}

\begin{acks}
  The authors thank the Middleware '25 reviewers for their valuable feedback.
  This work is supported by France 2030 Initiative under grant ANR-23-PECL-0004 (PEPR Cloud).
\end{acks}

\bibliographystyle{plainnat}
\bibliography{biblio.bib}

\iflong
\clearpage
\appendix
\onecolumn
\input{appendix-model}
\input{appendix-proofs}
\fi

\end{document}

%% file: abstract.tex
\begin{abstract}
  Parallel programs require software support to coordinate access to shared data.
  For this purpose, modern programming languages provide strongly-consistent shared objects.
  To account for their many usages, these objects offer a large API.
  However, in practice, each program calls only a small subset of the interface.
  Leveraging such an observation, we propose to tailor a shared object for a specific usage.
  We call this principle \emph{adjusted objects}.
  Adjusted objects already exist in the wild.
  This paper provides their first systematic study.
  We explain how everyday programmers already adjust common shared objects (such as queues, maps, and counters) for better performance.
  We present the formal foundations of adjusted objects using a new tool to characterize scalability, the indistinguishability graph.
  Leveraging this study, we introduce a library named \sys to inject adjusted objects in a Java program.
  In micro-benchmarks, objects from the \sys library improve the performance of standard JDK shared objects by up to two orders of magnitude.
  We also evaluate \sys with a Retwis-like benchmark modeled after a social network application. 
  On a modern server-class machine, \sys boosts by up to \speedup{1.7} the performance of the benchmark.
\end{abstract}

%% file: introduction.tex
\section{Introduction}
\labsection{introduction}

%% sharding is not stripping
%% stripping is the allocation of data request to a given thread
%% this changes from one execution of a datastore to another (because it depends on #threads)

%% Modern datastores are doing computation on shared data in stripes.
%% This means that threads are in charge of mutating a specific part of the data store.
%% A typical example is as follows:
%% Consider that the datastore is running with $n$ threads.
%% Then thread $t$ is doing all the client requests on keys $k=\modulo{n}{n}$.
%% In such a case, a key observation is that all the mutators on shared data are commuting operations.
%% We may leverage this by considering efficient implementation of such shared objects.

\paragraph{Context}
Over the past 20 years, the clock rate of processors has almost stalled.
From the \ghz{1} landmark at the turn of the 2000s, current processors just tick at a few GHz.
Nevertheless, Moore's law is well alive and the industry is able to always pack more transistors into the same square inch of silicon.
Manycore systems have 10s to 100s of hardware threads.
They are now common place in modern IT infrastructures, ranging from servers to small devices.
To harvest this computing power, programmers write parallel programs.
Ideally, these programs should be scalable, that is their performance improves linearly with the number of cores.

Getting right parallelism in an application is notoriously difficult.
To help programmers in this task, programming languages and libraries provide shared objects.
A plethora of objects is available, from
\begin{inparaenumorig}[]%
\item primitive data types (such as long, integer, reference),
\item collections (set, map, list, queue),
\item to complex synchronization mechanisms (compare-and-swap, lock, semaphore, monitor, barrier, transactions).
\end{inparaenumorig}
In most cases, these shared objects are strongly consistent.
This means that the interleaving of concurrent calls to the methods of an object is equivalent to a sequential execution.
Strong consistency allows the programmer to understand a shared object from its sequential specification.
This greatly simplifies the reasoning about correctness.

\begin{figure*}[!t]
  \hspace{-8em}
  \begin{minipage}{.4\textwidth}
    \begin{subfigure}{\textwidth}
      \scalebox{.65}{\input{figures-AtomicLong_stack_graph}}
      % \caption{\labfigure{mining:stack}}
    \end{subfigure} 
  \end{minipage}
  \hspace{8em}%
  \begin{minipage}{.25\textwidth}
    \begin{subfigure}{\textwidth}
      \vspace{1em}%
      \scalebox{.65}{\input{figures-AtomicLong_cassandra_graph}}
      % \caption{\labfigure{mining:tuple}}
    \end{subfigure}
  \end{minipage}
  ~~~~
  \caption{
    \labfigure{mining}
    Usage of \code{AtomicLong} across different open-source Java projects:
    \emph{(left)} lists the methods that are called.
    \emph{(right)} indicates if the returned value of a method is used (\textcolor{red}{+}), or not (\textcolor{blue}{$\times$}) in Apache Cassandra.
  }
\end{figure*}
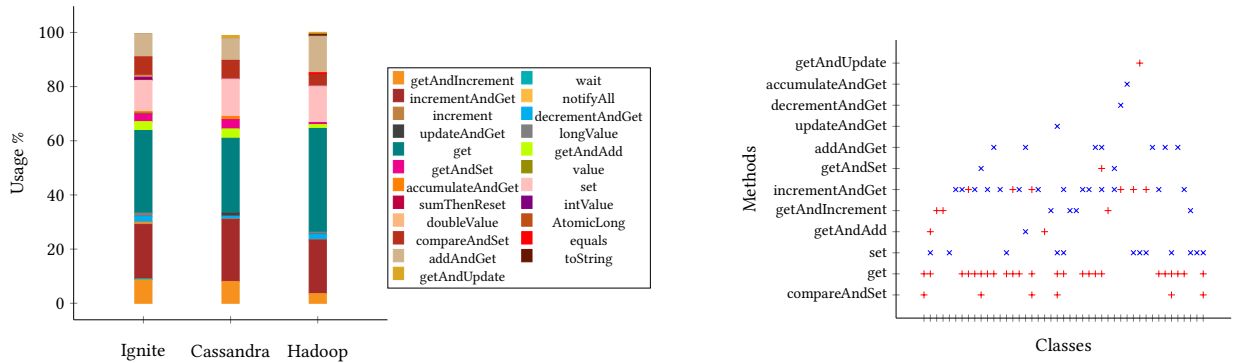

\paragraph{Motivation}
Typically, a shared object provides a wide interface with complex capabilities.
This accounts for the many usages the object has across the spectrum of applications.
Understanding how this interface is used in practice may provide hints to improve the object.

A first observation is that an application rarely uses the interface of a shared object in full.
To assess this point, we have delved into the usage of the \textit{java.util.concurrent} library in several prominent open-source projects.
\reffigure{mining} reports on our findings for \code{AtomicLong}.
In \reffigure{mining}(left), we can notice that
\emph{(i)} projects only use a handful of the available methods,
\emph{(ii)} some methods are used much more frequently than others, and
\emph{(iii)} certain groups of methods are employed specifically altogether.
Additionally, the return value of a method is not always used.
We illustrate this point in \reffigure{mining}(right).
For Apache Cassandra, this figure indicates if a class makes use of the return value, or not.
In many cases, the response of the method call is just ignored.

Another key observation is that if the specification is not used in full, it might be possible to implement a better version of the object.
Some expert programmers have already leveraged this insight to write objects ad-hoc to their use cases.
These implementations boost the performance of a program with improved scalability thanks to reduced hardware contention.
For instance, \reflst{linkedin} shows an extract of the Conccurentli project at Linkedin \cite{concurrentli}.
In this library, the programmers have implemented a specific version of \code{AtomicReference}.
Here, as the reference is written at most once (line~15), it is possible to buffer its value in the CPU cache for better efficiency (lines~6-7).
Other real-life examples of ad-hoc objects include counters \cite{LongAdder} and queues \cite{MpscLinkedQueue, FastSizeDeque}.

\paragraph{Objective \& Contributions}
In this paper, we propose to tailor a shared object for a specific use.
We call this principle \emph{adjusted objects}.
As pointed above, adjusted objects already exist in the wild.
This paper presents the first systematic study of them.
In detail, we claim the following contributions:
\begin{itemize}
\item
  We present the formal foundations of adjusted objects.
  First, we introduce a new characterization of the scalability of a shared object: the indistinguishability graph.
  In a nutshell, the denser this graph is the more scalable the object is.
  Adjusting an object is interpreted as adding more edges to the graph.
\item Based on this characterization, we then propose \sys a library of adjusted objects for Java.
  \sys includes different adjusted objects for primitive data types, collection, and synchronization mechanisms.
  These objects are mirroring the shared objects found in the Java Development Kit (JDK).
  They offer an efficient implementation tailored to a specific usage.
\item
  We evaluate the performance benefits of \sys across a range of micro-benchmarks.
  In these benchmarks, \sys improves the performance of the standard JDK objects by up to two orders of magnitude.
  We also evaluate \sys with a Retwis-like \cite{retwis} benchmark modeled after a social network application. 
  On a modern server-class machine, \sys boosts by up to \speedup{1.7} the performance of the benchmark.
\item
  We also report on the usage of shared objects across 50 projects of the Apache Software foundation.
  Our study assesses that these objects
  \emph{(i)} account for a small part of a code base,
  \emph{(ii)} they are central to an application, and
  \emph{(iii)} as underlined above, each program uses only a fraction of an object's interface.
\end{itemize}

\begin{lstlisting}[belowskip=-0.8 \baselineskip, float=t, label=lst:linkedin, caption={Extract from the Concurrentli library}, xleftmargin=3em, linewidth=22em]
public class AtomicWriteOnceReference<T> {
    private T _cachedObj = null;
    private volatile T _obj = null;    
    public T get() {
	if (_cachedObj != null) return _cachedObj;
	_cachedObj = _obj;
	return _cachedObj;
    }
    public boolean set(T value) {
        if (!trySet(value))
	    throw new IllegalStateException("already set");
        return true;
    }
    public boolean trySet(T value) {
        if (get() != null) return false;
        if (!UPDATER.compareAndSet(this, null, value)) return false;
        _cachedObj = value;
        return true;
    }    
    /*..*/
}
\end{lstlisting}
% \vspace{-.5em}

\paragraph{Outline}
\refsection{model} introduces our system model and definitions.
In \refsection{dist}, we present the indistinguishability graph and detail how this notion helps to measure scalability.
Based on such graphs, \refsection{adjusted} principles the notion of adjusted object.
The \sys library is described in \refsection{design}.
\refsection{evaluation} details our evaluation campaign.
We survey related work in \refsection{related} before closing in \refsection{conclusion}.
\iflong
For readability, proofs are deferred to the appendix.
\else
Due to space constraints, the proofs are deferred to the long version of this paper \cite{long-version}
\fi

%% file: figures-AtomicLong_stack_graph.tex
\begin{tikzpicture}
  \pgfplotsset{ every non boxed x axis/.append style={x axis line style=-},
    every non boxed y axis/.append style={y axis line style=-}}
  \pgfplotsset{every tick label/.append style={font=\Large}}
  \begin{axis}[ybar stacked, axis x line=bottom, axis y line=left, enlarge x limits=true, enlarge y limits=true, grid=minor, label style={font=\Large}, xlabel={}, ylabel={Usage \%},enlarge x limits=0.4,  width=8cm, height=8cm, legend columns=2, legend style={font=\LARGE, nodes={scale=0.7, transform shape}}, legend style={at={(1,0.1)},anchor=south west}, xtick=data, xticklabel style = {yshift=-0.3cm}, xticklabels={Ignite, Cassandra, Hadoop}, title={}]
\addplot[fill, color=BurntOrange] coordinates
    { (0,8.64) (1,8.05) (2,3.55) };
\addplot[fill, color=TealBlue] coordinates
    { (0,0.41) (1, 0.00) (2, 0.00) };
\addplot[fill, color=Brown] coordinates
    { (0,20.16) (1,22.99) (2,19.86) };
\addplot[fill, color=Dandelion] coordinates
    { (0,0.41) (1, 0.00) (2, 0.00) };
\addplot[fill, color=brown] coordinates
    { (0,0.41) (1, 0.00) (2, 0.00) };
\addplot[fill, color=cyan] coordinates
    { (0,2.06) (1,1.15) (2,2.13) };
\addplot[fill, color=darkgray] coordinates
    { (0, 0.00) (1,1.15) (2, 0.00) };
\addplot[fill, color=gray] coordinates
    { (0,1.23) (1, 0.00) (2,0.71) };
\addplot[fill, color=teal] coordinates
    { (0,30.45) (1,27.59) (2,38.3) };
\addplot[fill, color=lime] coordinates
    { (0,3.29) (1,3.45) (2,1.42) };
\addplot[fill, color=magenta] coordinates
    { (0,2.88) (1,3.45) (2,0.71) };
\addplot[fill, color=olive] coordinates
    { (0,0.41) (1, 0.00) (2, 0.00) };
\addplot[fill, color=orange] coordinates
    { (0,0.41) (1,1.15) (2, 0.00) };
\addplot[fill, color=pink] coordinates
    { (0,11.52) (1,13.79) (2,13.48) };
\addplot[fill, color=purple] coordinates
    { (0,0.41) (1, 0.00) (2, 0.00) };
\addplot[fill, color=violet] coordinates
    { (0,0.82) (1, 0.00) (2, 0.00) };
% \addplot[fill, color=yellow] coordinates
    % { (0, 0.00) (1,1.15) (2, 0.00) };
\addplot[fill, color=Apricot] coordinates
    { (0,0.41) (1, 0.00) (2, 0.00) };
\addplot[fill, color=Bittersweet] coordinates
    { (0,0.41) (1, 0.00) (2, 0.00) };
\addplot[fill, color=BrickRed] coordinates
    { (0,6.58) (1,6.9) (2,4.26) };
\addplot[fill, color=Red] coordinates
    { (0, 0.00) (1, 0.00) (2,0.71) };
\addplot[fill, color=Tan] coordinates
    { (0,8.64) (1,8.05) (2,13.48) };
\addplot[fill, color=Sepia] coordinates
    { (0, 0.00) (1, 0.00) (2,0.71) };
\addplot[fill, color=Goldenrod] coordinates
    { (0, 0.00) (1,1.15) (2,0.71) };
% \addplot[fill, color=GreenYellow] coordinates
    % { (0,0.41) (1, 0.00) (2, 0.00) };
% \legend{, , incrementAndGet, , , , , , get, , , , , set, , , , , , , , addAndGet, , , };

\legend{getAndIncrement, wait, incrementAndGet, notifyAll, increment, decrementAndGet, updateAndGet, longValue, get, getAndAdd, getAndSet, value, accumulateAndGet, set, sumThenReset, intValue, doubleValue, AtomicLong, compareAndSet, equals, addAndGet, toString, getAndUpdate};
\end{axis}
\end{tikzpicture}

%% file: figures-AtomicLong_cassandra_graph.tex
\begin{tikzpicture}
  \pgfplotsset{ every non boxed x axis/.append style={x axis line style=-},
    every non boxed y axis/.append style={y axis line style=-}}
  \pgfplotsset{every tick label/.append style={font=\normalsize}}
  \begin{axis}[scatter/classes={U={mark=+,red}, NU={mark=x,blue}}, label style={font=\Large}, legend pos=outer north east,axis x line=bottom, axis y line=left, enlarge x limits=true, xlabel = {Classes}, ylabel = {Methods}, enlarge y limits=true, xtick = data, xticklabels = {,,},ytick=data, yticklabels={compareAndSet,get,set,getAndAdd,getAndIncrement,incrementAndGet,getAndSet,addAndGet,updateAndGet,decrementAndGet,accumulateAndGet,getAndUpdate}, title={}]
\addplot[scatter,only marks, scatter src=explicit symbolic]
coordinates {
(0,0) [U]
(0,1) [U]
(1,2) [NU]
(1,3) [U]
(1,1) [U]
(2,4) [U]
(3,4) [U]
(4,2) [NU]
(5,5) [NU]
(6,5) [NU]
(6,1) [U]
(7,1) [U]
(7,5) [U]
(8,5) [NU]
(8,1) [U]
(9,0) [U]
(9,6) [NU]
(9,1) [U]
(10,5) [NU]
(10,1) [U]
(11,7) [NU]
(11,1) [U]
(12,5) [NU]
(13,2) [NU]
(13,1) [U]
(14,1) [U]
(14,5) [U]
(15,5) [NU]
(15,1) [U]
(16,3) [NU]
(16,7) [NU]
(17,0) [U]
(17,1) [U]
(17,5) [U]
(18,5) [NU]
(19,3) [U]
(20,4) [NU]
(21,0) [U]
(21,2) [NU]
(21,8) [NU]
(21,1) [U]
(22,5) [NU]
(22,2) [NU]
(22,1) [U]
(23,4) [NU]
(24,4) [NU]
(25,5) [NU]
(25,1) [U]
(26,5) [NU]
(26,1) [U]
(27,7) [NU]
(27,1) [U]
(28,5) [NU]
(28,6) [U]
(28,7) [NU]
(28,1) [U]
(29,4) [U]
(30,5) [NU]
(30,6) [NU]
(31,9) [NU]
(31,5) [U]
(32,10) [NU]
(33,2) [NU]
(33,5) [U]
(34,2) [NU]
(34,11) [U]
(35,2) [NU]
(35,5) [U]
(36,7) [NU]
(37,5) [NU]
(37,1) [U]
(38,1) [U]
(38,7) [NU]
(39,0) [U]
(39,2) [NU]
(39,1) [U]
(40,7) [NU]
(40,1) [U]
(41,5) [NU]
(41,1) [U]
(42,4) [NU]
(42,2) [NU]
(43,2) [NU]
(44,0) [U]
(44,2) [NU]
(44,1) [U]
};
\end{axis}
\end{tikzpicture}

%% file: model.tex
\section{System Model}
\labsection{model}

\iflong
In what follows, we give a brief overview of our system model.
The full details are in \refappendix{model}.
\fi

% basics
We assume the standard shared memory model of computation.
In this model, a set of $n \geq 2$ threads communicate using shared objects.
Each object abides by a sequential specification, or \emph{data type}.
For an object $O$, $O.T$ denotes its data type.
The base objects of the shared memory model are registers to which threads read and write.
The memory also offers common synchronization primitives, such as compare-and-swap and test-and-set, to update a register.
% wait-freedom; linearizability
Registers are used to implement higher-level objects (such as maps, queues, and trees).
We restrict our attention to wait-free linearizable deterministic objects.
Wait-freedom guarantees that every call to the interface eventually returns.
Linearizability ensures that the interleaving of concurrent calls is equivalent to a sequential execution that preserves real time.

% restrictions
Some objects may have restricted use.
They can be one-shot, that is each thread calling them at most once---in contrast with long-lived objects.
Another restriction is what part of the interface is available to which thread.
For instance, single-writer multiple-readers ($\SWMR$) registers are common in literature \cite{art}.
To model this, we consider that each object $O$ has an \emph{access permission map} $O.m$ that defines operations a thread can execute.
Every operation defined by the data type $O.T$ is executable by at least one thread.

% parallelism
Threads may access a common register when they call a shared object.
If this happens and the register is updated, we call it a conflict.
In detail, operations $c$ and $d$ called by two distinct threads \emph{conflict} when they access a common register that one of them updates.
Operations $c$ and $d$ \emph{update conflict} when they both update a common register \cite{bayou}.
Conversely, an operation is \emph{invisible} when it never updates a register accessed by some other operation \cite{invisible}.
An operation is \emph{conflict-free} when it never conflicts with any other operation.
When no two operations conflict, the implementation is \emph{conflict-free}.
In such a situation, shared data is read-only.
In the case where threads do not access a common register, the implementation is disjoint-access parallel \cite{dap}.

% scalability
An implementation \emph{scales} when its performance grows with the number of threads.
In practice, several parameters influence scalability, such as the size and locality of the working set.
When the workload is parallel-enough, the key factor is the conflict rate.
On modern hardware, an implementation scales when conflicts are rare \cite{rule}.
Hence, it is paramount to understand when such an implementation is feasible.

%% file: distinguishability.tex
\section{Scalability}
\labsection{dist}

To principle adjusted objects, we need a good indicator of scalability.
Because it captures the synchronization power, the consensus number is a potential candidate in this task.
Unfortunately, it is well-known that this is not the case.
We start this section by recalling such a result (\refsection{dist:motivation}).
To fill the gap, we then propose a new metric, the indistinguishability graph (\refsection{dist:def}).
The indistinguishability graph is built from the sequential data type of a shared object.
We explain how scalability is closely related to the density of this graph (\refsection{dist:scalability}).

\subsection{On the consensus number}
\labsection{dist:motivation}

In distributed computing, a central object is consensus.
Its interface exposes a single operation: $\propose(v)$, where $v$ is some input value.
Threads that call $\propose()$ all return (aka. \emph{agree} on) one of the input values.

Given a data type $T$, the maximal number of threads that may agree using objects of type $T$ and read/write operations on registers is called the consensus number of $T$ \cite{Herlihy}.
We write $\cn(T)$ the consensus number of $T$, and $\cn_k$ all the shared objects with consensus number $k$.
We will say that an operation $c$ of $T$ has consensus power when $T$ restricted to just operation $c$ has a consensus number greater than 1.

Objects like queues permit to reach consensus among several threads.
Registers alone are too weak to solve a distributed agreement.
They belong to $\cn_1$.

A $m$-valued max-register stores up to $log(m)$ values.
Updating its content demands $O(min(n,log(m)))$ read/write operations on the registers \cite{aspnes12}.
Differently, reading a snapshot object requires $\Omega(n)$ such steps \cite{SnapshotComplexity}.
Hence, despite these two objects are in $\cn_1$, they scale differently.
This explains why the consensus number is a poor indicator of scalability.

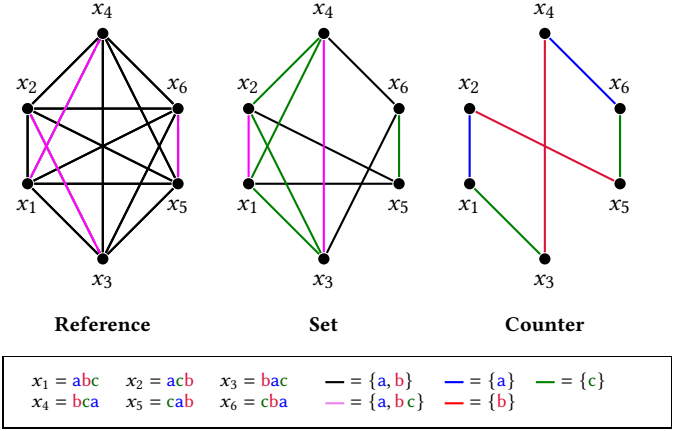
\begin{figure}[!t]
  \centering
  \captionsetup{justification=centering}
  \begin{subfigure}[t]{.15\textwidth}
    \centering
    \input{figures-ref.tex}
  \end{subfigure}
  \hfill
  \begin{subfigure}[t]{.15\textwidth}
    \centering
    \input{figures-set.tex}
  \end{subfigure}
  \hfill
  \begin{subfigure}[t]{.15\textwidth}
    \centering
    \input{figures-cnt.tex}
  \end{subfigure}
  \hfill
  \begin{subtable}[t]{.5\textwidth}
    \footnotesize
    \bigskip
    \centering
    \begin{tcolorbox}[
        sharp corners,
        boxrule=0.5pt,
        colback=white,
        boxsep=-.5pt,
        colframe=black
      ]
      \hspace{-1em}
      \begin{tabular}{lll}
        $x_1 = \opa \opb \opc$ &
        $x_2 = \opa \opc \opb$ &
        $x_3 = \opb \opa \opc$ \\
        $x_4 = \opb \opc \opa$ &
        $x_5 = \opc \opa \opb$ &
        $x_6 = \opc \opb \opa$
      \end{tabular}
      \quad
      \begin{tabular}{@{}ll}
        ${\color{Black}{\keyline}} = \{\opa, \opb\}$  \\
        ${\color{Violet}{\keyline}} = \{\opa, \opb\, \opc\}$
      \end{tabular}
      ~
      \begin{tabular}{@{}l}
        ${\color{Blue}{\keyline}} = \{\opa\}$ \\
        ${\color{Red}{\keyline}} =\{\opb\}$
      \end{tabular}
      ~
      \begin{tabular}{@{}l@{}}
        ${\color{Green}{\keyline}} = \{\opc\}$ \\
        ~
      \end{tabular}
    \end{tcolorbox}
  \end{subtable}
  \caption{
    \labfigure{ig}
    From left to right, the indistinguishability graph $\Gr(\{\opa,\opb,\opc\})$ for 
    a reference ($\opa=\set(1)$, $\opb=\set(2)$, $\opc=\get()$),
    a set ($\opa=\add(1)$, $\opb=\add(1)$, $\opc=\contains(1)$), and
    a counter ($\opa=\inc(1)$, $\opb=\inc(3)$, $\opc=\inc(5)$).
  }
  % \vspace{-.5em}%
\end{figure}

\subsection{Indistinguishability graph}
\labsection{dist:def}

This section introduces the indistinguishability graph.
Such a graph is built from the data type of the shared object.
We later explain how scalability and the density of the indistinguishability graph are related.

\paragraph{Definition}
Consider an object of type $T$, a state $s$, and a bag (a set with possible repetitions) $B$ of operations.
For some permutation $x$ of $B$, the response of $c \in B$ from $s$ is the return value of $c$ after applying $x$ until $c$ from state $s$. % FIXME beware occurences!
Two permutations $x$ and $x'$ are \emph{indistinguishable} from $s$ for $c$, denoted $x \indistinguishable{c,s} x'$, when
\begin{inparaenum}%
\item $c$ obtains the same response in both permutations, and
  % \tau(s,\until{x}{c}).\val = \tau(s,\until{x'}{c}).\val
\item there exists a common state $s'$ attainable after $c$ in both $x$ and $x'$.
  % \exists s', \until{x}{c} \pref \hat{x} \pref x, \until{x'}{c} \pref \hat{x'} \pref x', \tau(s,\hat{x}).\st = \tau(s,\hat{x'}).st
\end{inparaenum}%
% FIXME does 'or' imply consensus not solvable?
% Notice that this relation is clearly symmetric.
Hereafter, when $c$ and $s$ are clear from the context, we omit them ($x \indist x'$).

A convenient way to represent relation $\indist$ is through an \emph{indistinguishability graph}.
Given $B$ and $s$, the nodes of this graph $\Gr_{T}(B,s)$ corresponds to all the permutations of $B$.
There exists an edge $(x,x')$ \emph{labeled} with $C$ in $\Gr_{T}(B,s)$ if and only if $x$ and $x'$ are indistinguishable from $s$ for all the operations $c \in C$.
The label is considered \emph{strong} if applying $x$ from $s$ results in the same state as applying $x'$.

The \emph{indistinguishability class} of $x$ is the transitive closure of $x$ by relation $\indist$.
It is written $\clazz{x}^{B}_s$.
This corresponds to the connected component $x$ belongs to in $\Gr_{T}(B,s)$.
When $B$ and $s$ are clear from the context, we just write $\clazz{x}$.
An operation $c$ is \emph{labeling} in $\Gr_{T}(B,s)$ when it is a label of all the edges.
In this case, there is one indistinguishability class.
When $c$ is a strong label of all the edges, it is \emph{strongly labeling}.

These notions extend naturally to shared objects.
In detail, let $O$ be an object of type $T$.
The indistinguishability graphs of $O$ are the graphs $\Gr_T(B,s)$ such that
$s$ is a state of $O$, and
$B$ complies with the access permission map $O.m$ (that is, each thread $p$ is mapped to a unique operation $c_p$ in $B$ with $c_p \in O.m[p]$).
The indistinguishability graphs of $O$ are written $\Grs_{O}$.

\input{figures-spec}

\paragraph{Examples}
\reffigure{ig} illustrates the indistinguishability graph of three common objects: a reference, a set, and a counter.
There are three operations $\{\opa,\opb,\opc\}$, applied from the initial state.
The formal specification of these objects and operations appear in \reftab{spec} (first column).

In total, there are $3!=6$ nodes in the indistinguishability graph.
For a reference (left side of \reffigure{ig}), the graph is complete because $\set$ does not return anything.
Hence all edges have (at least) the default label $l=\{\opa,\opb\}$.
When the same $\set$ is immediately preceding $\get$, it returns the same response.
Thus, $\opc$ labels the edges ($x_1,x_4$), ($x_2,x_3$), and ($x_5,x_6$).

For a set object, $\opc=\contains(1)$ is labeling when it is not the first operation.
Indeed, in such cases it always returns $\true$.
Moreover, when the $\add(1)$ operations are in the same order, their responses do not change.
In those cases, $\opa$ and $\opb$ are labeling.
Whatever the permutation is, the set always ends up in the same final state.
Hence all labels are strong in this example.

The rightmost figure illustrates three increments of $1$, $3$, and $5$ applied to a counter.
Each increment returns the state of the counter after it is applied.
Thus, if we permute the first two operations, the last operation will return the same value.

In \reffigure{ig}, all the graphs are connected.
As a consequence, there is a single indistinguishability class.
In general, there are at most $\cardinalOf{B}$ indistinguishability classes.
This comes from the fact that if $x[0]=y[0]$ then $\clazz{x} = \clazz{y}$.

\paragraph{Relation with consensus}
When among $k=\cardinalOf{B}$ operation, a data type $T$ has the power to distinguish up to $l$ classes, we say that $T$ is $\dist{k,l}$.
% It is easy to see that if $T$ is $\dist{k,l}$ it is at most $\dist{k,l+1}$.
%
For instance, an increment-only counter is $\dist{2,2}$ but only $\dist{3,1}$, as the third operation cannot distinguish how the preceding ones are ordered (see \reffigure{ig}).
This is also the case with a set.
In these examples, key is the transition to $\dist{k,1}$.
Below, we establish that when such a transition exists and $T$ is readable, that is any thread may retrieve the latest state of the object, $k$ is the consensus number of $T$.
\iflong
The proof is in \refappendix{proofs:consensus}.
This appendix also contains a full characterization of $\cn_1$.
\fi

\begin{restatable}{thm}{distcn}
  \labtheo{distcn}
  Consider a readable data type $T$.
  Then,
  $
  \cn(T)
  =
  \max
  ~\{ k : \exists l \geq 2 \sep T \in \dist{k,l} \}
  \union
  \{ 1 \}
  $
\end{restatable}

\subsection{Predicting scalability}
\labsection{dist:scalability}

We now connect the scalability of a shared object $O$ to the shape of its indistinguishability graphs $\Grs_O$.
First, we provide necessary and sufficient conditions on those graphs for a conflict-free implementation of $O$ to exist.
Then, we turn our attention to updates conflicts and invisible operations.
\iflong
\refappendix{proofs:scalability} contains the proofs of the results in this section.
\fi

\paragraph{Conflict freedom}
% FIXME should be about one operation instead?
% explain that commuting writes belong to that category
% get from the first set is labeling (in fact, one can copy the state localy and from there..)
For starters, we investigate one-shot objects, that is objects callable at most once per thread.
Recall from \refsection{model} that an implementation is conflict-free when no two operations access the same register, with one of them updating it.
We can establish the following result.

\begin{restatable}{prop}{oneshot}
  \labprop{oneshot}
  Assume that object $O$ is one shot.
  There exists a conflict-free implementation of $O$
  if and only if
  $B$ is labeling in every $\Gr_T(B,s)$ of $\Grs_O$.
\end{restatable}

In \cite{rule}, the authors introduce the SIM-commutativity rule.
The rule indicates when an interval of a concurrent execution can be implemented in a conflict-free manner.
This corresponds to the sufficient part of \refprop{oneshot}:
it happens when all operations $B$ in the interval are labeling in $\Gr_T(B,s)$, where $s$ is the state of the object before the interval.
\refprop{oneshot} proves that it is also necessary for deterministic objects.
This result extends to the case of long-lived objects as detailed below.

\begin{restatable}{prop}{longlived}
  \labprop{longlived}
  A conflict-free implementation of $O$ exists
  if and only if
  $B$ is strongly labeling in every $\Gr_T(B,s)$ of $\Grs_O$ with $\cardinalOf{B}=2$.
\end{restatable}

The above condition demands that any two operations of different threads commute.
This may happen when they access different shards, or \emph{segments}, in a large object.
For instance, in a key-value store, a thread can be in charge of a particular key range of the store.
%% In this case, it is possible to do a disjoint-access parallel implementation of the object:
%% shared data, which is necessarily read-only, is de-duplicated.

In \reflst{linkedin}, the \code{AtomicWriteOnceReference} class does not satisfy the premises of \refprop{longlived} when $B=\{\code{set},\code{get}\}$.
Nonetheless, this data type scales in practice (see \refsection{evaluation}).
One observation is that its indistinguishability graph is dense.
Indeed, permuting operations before (or after) the first \code{set} does not change their return values, nor the state of the object.
Our next results further delve into the relation between the shape of the indistinguishability graphs and scalability.

\paragraph{Update conflicts}
Scalability relates to the existence of conflicts among operations.
Roughly speaking, a conflict is necessary when two operations impact each other.
To track this, we introduce a moverness semantic \cite{lipton75,koskinen10}.

Consider an indistinguishability graph $\Gr_T(B,s)$ and a permutation $x=c_1{\ldots}c_{m \geq 2}$ of $B$.
An operation $c_i$ \emph{left-moves} in $x$ when $c_i$ strongly labels the edge $(x,x')$ in $\Gr_T(B,s)$, with $x'=c_1{\ldots}c_{i}c_{i-1}{\ldots}c_m$.
By extension, the operation left-moves in $\Gr_T(B,s)$ when it left-moves in all permutations of $B$.
If it left-moves in all the indistinguishability graphs of the object, we say that the operation is a \emph{left-mover}.

Blind writes, that is writes whose return value is null, are potential left-movers.
For instance, consider a set object.
If $\add$ is blind (object $S_2$ in \reftab{spec}), it left-moves with prior $\add$ operations.
Another example is $\offer$ for a queue (object $Q_1$ in \reftab{spec}).
When the queue is not empty, this operation left-moves with $\poll$.
Hence, if at most one thread invokes $\offer$, the operation is likely to left-move.
This is an example of adjusted object and \refsection{design} covers it in detail.
The proposition below states that left-movers can be implemented efficiently.

\begin{restatable}{prop}{updateconflict}
  Operation $c$ is implementable without update conflicts
  if
  $c$ left-moves in every $\Gr_T(B,s)$ of $\Grs_O$.
  Such an implementation exists only if $c$ has no consensus power.
\end{restatable}

\paragraph{Invisible operations}
Right-movers are another type of operations implementable in an efficient manner.
Such operations do not change the result of an adjacent operation when their positions are swapped.
More precisely, $c_i$ right-moves in $x=c_1{\ldots}c_{i-1}c_i{\ldots}c_m$ when $c_{i-1}$ strongly labels the edge $(x,x')$ in $\Gr_T(B,s)$, where $x'=c_1{\ldots}c_{i}c_{i-1}{\ldots}c_m$.
Notice that $c_i$ right-moves in $x$ if and only if $c_{i-1}$ left-moves in $x'$.
Right-movers are defined similarly to left-movers.

Because they have no side effects, reads are typical right-movers.
However, they are not the only one.
This is also likely to happen when updates are blind.
% For instance, with a single-reader implementation of the counter $C_3$ in \reftab{spec}.
The result below proves that right-movers are implementable in an invisible manner.

\begin{restatable}{prop}{invisible}
  If $c$ right-moves in every $\Gr_T(B,s)$ of $\Grs_O$,
  then there is an implementation of $c$ in which it is invisible.
\end{restatable}

%% To better understand how the indistinguishability graph can guide design choices and improve scalability, we explore its applications in selecting, comparing, and updating data types within an application.

\subsection{How to use the indistinguishability graph?}
\labsection{dist:usage}

Programmers can leverage the indistinguishability graph to write better concurrent code.

For instance, \reftheo{distcn} states that if the indistinguishability graph is not connected and the object is readable, then it has a consensus number higher than $1$.
As a consequence, any implementation requires internally some form of synchronization \cite{Herlihy16}.

The graph also helps with comparing shared objects.
As an example, consider that we have the choice between objects $O$ and $O'$.
We can use the indistinguishability graph to decide:
if the indistinguishability graphs of $O$ include the ones of $O'$, then it is likely that $O$ scales better.
For instance, this happens when we compare \code{AtomicWriteOnceReference} (presented in \reflst{linkedin}) with \code{AtomicReference} from the JDK.

Building upon such observations, the next logical step is to adjust an object to a specific usage in a given piece of code.
We further delve into this idea in the next section.

%% file: figures-ref.tex
\begin{tikzpicture}
  [
    scale=1,
    auto,
    shifttl/.style={shift={(-\shiftpoints,\shiftpoints)}},
    shifttr/.style={shift={(\shiftpoints,\shiftpoints)}},
    shiftbl/.style={shift={(-\shiftpoints,-\shiftpoints)}},
    shiftbr/.style={shift={(\shiftpoints,-\shiftpoints)}},
  ]
  \begin{scope}[<-,thick]
    \node[adot,label=270:$x_1$] (v1) at (0,0) {};
    \node[adot,label=270:$x_3$] (v3) at (1,-1) {};
    \node[adot,label=270:$x_5$] (v5) at (2,0) {};
    \node[adot,label=$x_2$] (v2) at (0,1) {};
    \node[adot,label=$x_4$] (v4) at (1,2) {};
    \node[adot,label=$x_6$] (v6) at (2,1) {};

    \foreach \myx in {1,...,6}{
      \foreach \myy in {1,...,6}{
        \path[basic] (v\myx) edge (v\myy);
      }
    }

    % remove some edges
    \draw[draw, -latex', White] (v1) -- (v4); 
    \draw[draw, -latex', White] (v5) -- (v6);
    \draw[draw, -latex', White] (v2) -- (v3); 
    
    \draw[ab] (v1) -- (v4);
    \draw[ab] (v2) -- (v3);
    \draw[ab] (v5) -- (v6);
  \end{scope}
\end{tikzpicture}
\caption*{Reference} % \labfigure{ig:reg}

%% file: figures-set.tex
\begin{tikzpicture}
  [
    scale=1,
    auto,
    shifttl/.style={shift={(-\shiftpoints,\shiftpoints)}},
    shifttr/.style={shift={(\shiftpoints,\shiftpoints)}},
    shiftbl/.style={shift={(-\shiftpoints,-\shiftpoints)}},
    shiftbr/.style={shift={(\shiftpoints,-\shiftpoints)}},
  ]
  \begin{scope}[<-,thick]
    \node[adot,label=270:$x_1$] (v1) at (0,0) {};
    \node[adot,label=270:$x_3$] (v3) at (1,-1) {};
    \node[adot,label=270:$x_5$] (v5) at (2,0) {};
    \node[adot,label=$x_2$] (v2) at (0,1) {};
    \node[adot,label=$x_4$] (v4) at (1,2) {};
    \node[adot,label=$x_6$] (v6) at (2,1) {};
    % \path[basic,label=$\opa\opb$] (v1) edge (v2);
    \path[basic] (v1) edge (v5);
    \path[basic] (v2) edge (v5);
    \path[basic] (v3) edge (v6);
    \path[basic] (v4) edge (v6);

    \draw[basic, draw=Green] (v2) -- (v4);
    % node[midway, above, yshift=1ex, xshift=-1ex] (label1) {$\opc$}
    % (label1) edge[thin] ++(0.2,-0.3); 

    \draw[basic, draw=Green] (v1) -- (v4);
    % node[near end, below, yshift=-1ex, xshift=0.5ex] (label2) {$\opc$}
    % (label2) edge[thin] ++(-0.08,0.1); 

    \draw[basic, draw=Green] (v2) -- (v3);
    % node[midway, above, yshift=0.5ex, xshift=1ex] (label3) {$\opc$}
    % (label3) edge[thin] ++(-0.3,0); 

    \draw[basic, draw=Green] (v1) -- (v3);
    % node[midway, below, yshift=-1ex, xshift=-1ex] (label4) {$\opc$}
    % (label4) edge[thin] ++(0.2,0.3); 

    \draw[ab] (v1) -- (v2);
    % node[midway, left, xshift=-1ex] (label5) {$\opa\opb\opc$}
    % (label5) edge[thin] ++(0.51,0); 

    \draw[ab] (v3) -- (v4);
    % node[near end, right, xshift=0.1ex] (label6) {$\opa\opb\opc$}
    % (label6) edge[thin] ++(-0.25,0); 

    \draw[basic, draw=Green] (v5) -- (v6);
    % node[midway, right, xshift=0.1ex] (label7) {$\opc$}
    % (label7) edge[thin] ++(-0.1,0); 
    
  \end{scope}
\end{tikzpicture}
\caption*{Set} % \labfigure{ig:set}

%% file: figures-cnt.tex
\begin{tikzpicture}
  [
    scale=1,
    auto,
    shifttl/.style={shift={(-\shiftpoints,\shiftpoints)}},
    shifttr/.style={shift={(\shiftpoints,\shiftpoints)}},
    shiftbl/.style={shift={(-\shiftpoints,-\shiftpoints)}},
    shiftbr/.style={shift={(\shiftpoints,-\shiftpoints)}},
  ]
  \begin{scope}[<-,thick]
    \node[adot,label=270:$x_1$] (v1) at (0,0) {};
    \node[adot,label=270:$x_3$] (v3) at (1,-1) {};
    \node[adot,label=270:$x_5$] (v5) at (2,0) {};
    \node[adot,label=$x_2$] (v2) at (0,1) {};
    \node[adot,label=$x_4$] (v4) at (1,2) {};
    \node[adot,label=$x_6$] (v6) at (2,1) {};

    \draw[basic, draw=Blue] (v1) -- (v2);
    % % node[midway, left, xshift=-1ex] (label1) {$\opa$}
    % % (label1) edge[thin] ++(0.38,0); 

    \draw[basic, draw=Green] (v1) -- (v3);
    % % node[midway, below, yshift=-1ex, xshift=-1ex] (label2) {$\opc$}
    % % (label2) edge[thin] ++(0.2,0.3);

    \draw[basic, draw=Crimson] (v2) -- (v5);
    % % node[near start, below, yshift=-0.2ex] (label3) {$\opb$}
    % % (label3) edge[thin] ++(0.1,0.08); 
    
    \draw[basic, draw=Crimson] (v3) -- (v4);
    % % node[near end, right, xshift=0.1ex] (label4) {$\opb$}
    % % (label4) edge[thin] ++(-0.1,0);

    \draw[basic, draw=Blue] (v4) -- (v6);
    % % node[midway, above, xshift=0.1ex] (label5) {$\opa$}
    % % (label5) edge[thin] ++(-0.1,0);
    
    \draw[basic, draw=Green] (v5) -- (v6);
    % % node[midway, right, xshift=0.1ex] (label6) {$\opc$}
    % % (label6) edge[thin] ++(-0.1,0); 
    
  \end{scope}
\end{tikzpicture}
\caption*{Counter} % \labfigure{ig:cnt}

%% file: figures-spec.tex
%% \begin{itemize}
%% \item inc() : $\mathbb{N}$ or $\varnothing$
%% \item get() : $\mathbb{N}$
%% \item seal() : $\varnothing$
%% \end{itemize}

\begin{table*}
  \centering
  \footnotesize
  \begin{center}
    \begin{tabular}{@{}c@{~}|@{~}l@{~}|@{~}l@{~}|@{~}l@{}}
      Counter 
      &
      \begin{minipage}{20em}
        \hoare{\true}{\rmw(f,x)}{s' = f(s,x) \land \ret = s'} \\
        \hoare{\true}{\inc()}{s'  = s+1 \land \ret = s' }\\ 
        \hoare{\true}{\get()}{\ret = s} \\
        \hoare{\true}{\reset()}{s' = 0} \hfill $C_1$
        \smallskip
      \end{minipage}
      &
      \begin{minipage}{18em}
        \hoare{\true}{\rmw(f,x)}{\true} \\
        \hoare{\true}{\inc()}{s' = s+1 \land \ret = s' }\\
        \hoare{\true}{\get()}{\ret = s} \\
        \hoare{\false}{\reset()}{s' = 0} \hfill $C_2$
        \smallskip
      \end{minipage}
      &
      \begin{minipage}{14.5em}
        \hoare{\true}{ \rmw(f,x)}{\true} \\
        \hoare{\true}{\inc()}{s' = s+1} \\
        \hoare{\true}{\get()}{\ret = s} \\
        \hoare{\false}{\reset()}{s' = 0} \hfill $C_3$
        \smallskip
      \end{minipage}
      \\
      \hline
      \hline
      Set
      &
      \begin{minipage}{20em}
        \smallskip
        \hoare{\true}{\add(x)}{s' = s \union \{x\} \land \ret = x \notin s }  \\
        \hoare{\true}{\remove(x)}{s' = s \setminus \{x\} \land \ret = x \in s } \\
        \hoare{\true}{\contains(x)}{\ret = x \in s} \hfill $S_1$
        \smallskip
      \end{minipage}
      &
      \begin{minipage}{18em}
        \smallskip
        \hoare{\true}{\add(x)}{s' = s \union \{x\} } \\
        \hoare{\true}{\remove(x)}{s' = s \setminus \{x\} } \\
        \hoare{\true}{\contains(x)}{\ret = x \in s} \hfill $S_2$
        \smallskip
      \end{minipage}
      &
      \begin{minipage}{14.5em}
        \smallskip
        \hoare{\true}{\add(x)}{s' = s \union \{x\} } \\
        \hoare{\true}{\remove(x)}{\true} \\
        \hoare{\true}{\contains(x)}{\ret = x \in s} \hfill $S_3$
        \smallskip
      \end{minipage}
      \\
      \hline
      \hline
      %% List
      %% &
      %% \begin{minipage}{20em}
      %%   \smallskip
      %%   \hoare{\true}{\add(x)}{s' = s \append x \land \ret = x \notin s} \\        
      %%   \hoare{\true}{\remove(x)}{s' = s \setminus \{\first(x)\} \land \ret = x \in s} \\ % FIXME first occurence
      %%   \hoare{\true}{\contains(x)}{\true} \hfill $L_1$
      %%   \smallskip
      %% \end{minipage}
      %% &
      %% \begin{minipage}{20em}
      %%   \smallskip
      %%   \hoare{\true}{\add(x)}{s' = s \append x} \\
      %%   \hoare{\true}{\remove(x)}{s' = s \setminus \{\first(x)\}} \\
      %%   \hoare{\true}{\contains(x)}{\true} \hfill $L_2$
      %%   \smallskip
      %% \end{minipage}
      %% \\
      %% \hline
      %% \hline
      Queue
      &
      % \multicolumn{2}{l}{
        \begin{minipage}{20em}
          \smallskip
          \hoare{\true}{\offer(x)}{s' = s \append x} \\
          \hoare{\true}{\poll()}{\aif~\cardinalOf{s} = 0~\athen~\ret = \bot~\aelse~\ret = \head(s)~\land~s' = s \setminus \{\head(s)\}} \\
          \hoare{\true}{\contains(x)}{\ret = x \in s} \hfill $Q_1$
          \smallskip
        \end{minipage}
      % }
      %% &
      %% \begin{minipage}{20em}
      %%   \smallskip
      %%   \hoare{\true}{\offer(x)}{s' = s \append x} \\
      %%   \hoare{\cardinalOf{s} \geq 1}{\poll()}{s' = s \setminus \{head(s)\}}\\
      %%   \hoare{\true}{\contains(x)}{\ret = x \in s} \hfill $Q_2$
      %%   \smallskip
      %% \end{minipage}
      \\
      \hline
      \hline
      Reference
      &
      \begin{minipage}{20em}
        \smallskip
        \hoare{x \in \addrSet}{\set(x)}{s' = x} \\
        \hoare{\true}{\get()}{\ret = s} \hfill $R_1$
        \smallskip
      \end{minipage}
      &
      \begin{minipage}{18em}
        \smallskip
        \hoare{x \in \addrSet \land s = \bot}{\set(x)}{s' = x} \\
        \hoare{\true}{\get()}{\ret = s} \hfill $R_2$
        \smallskip
      \end{minipage}
      \\
      \hline
      \hline
      Map
      &
      \begin{minipage}{20em}
        \smallskip
        \hoare{\true}{\aput(k,v)}{s'[k] = v \land \ret = s[k]}  \\
        \hoare{\true}{\remove(k)}{s'[k] = \bot \land \ret = s[k]} \\
        \hoare{\true}{\contains(k)}{\ret = (s[k] \neq \bot)} \hfill $M_1$
        \smallskip
      \end{minipage}
      &
      \begin{minipage}{18em}
        \smallskip
        \hoare{\true}{\aput(k,v)}{s'[k] = v}  \\
        \hoare{\true}{\remove(k)}{s'[k] = \bot} \\
        \hoare{\true}{\contains(k)}{\ret = (s[k] \neq \bot)} \hfill $M_2$
        \smallskip
      \end{minipage}
      %% &
      %% \begin{minipage}{18em}
      %%   \smallskip
      %%   \hoare{\true}{\aput(k,v)}{s'[k] = v}  \\
      %%   \hoare{\true}{\remove(k)}{s'[k] = \bot} \\
      %%   \hoare{\true}{\contains(k)}{\true} \hfill $M_3$
      %%   \smallskip
      %% \end{minipage}
    \end{tabular}       
  \end{center}
  \caption{
    \labtab{spec}
    Adjusted versions of common data types in Hoare logic.
    We note $s$ the state of the object, $s'$ the state after the operation is applied, and $\ret$ the return value.
    If left unspecified, the state is unchanged and the return value is empty ($\bot$).
  }
  % \vspace{-2em}
\end{table*}

%% file: adjusted.tex
% \vspace{-1em}%
\section{Adjusted Objects}
\labsection{adjusted}

Adjusted object are objects tailored for a specific use in a program.
They are more scalable than their vanilla counterparts because they have denser indistinguishability graphs.
Below, we define adjusted objects and establish this result.
We also introduce a methodology to adjust an object by controlling how threads access it and/or altering its interface.

% \vspace{-1em}%
\subsection{Principles}
\labsection{adjusted:principle}

\paragraph{Definition}
In many situations, access to a shared object is not symmetrical.
A base example is a single-writer multi-reader (SWMR) register.
A unique thread (the writer) may invoke operation $\awrite$ on the register.
The other threads are readers and can only call $\aread$.
Another common situation is a queue with one producer and one consumer.
The producer (resp., consumer) only invokes $\offer$ (resp., $\poll$).
In such cases, the object is adjusted appropriately.
In detail, recall that we write $O.m$ the access permission map of a shared object $O$.
With this adjustment, $O.m$ is restricted to how the threads use operations in the specific use case.
For instance, a SWMR register has a single thread $w$, the writer, with $\awrite \in O.m[w]$.
Every other thread $p$ satisfies $O.m[p] = \{\aread\}$.

A second adjustment is to alter the interface.
For instance, as shown in \reffigure{mining}, the return value of a call can be ignored in part of a program.
% That is, the operation is blind in these use cases.
Another situation is when the interface is not used in full, but narrowed to a handful of operations.
All these modifications are captured with subtyping.

To formally define adjusted objects, we follow the work of \citet{subtyping} that introduces types and subtypes.
Recall that a data type $S$ is a subtype of $T$ when it satisfies the following two conditions:
there is an abstraction function that maps a state of $S$ to a state of $T$ which maintains state invariants; and
$S$ preserves the operations of $T$.
We say that $S$ is a \emph{narrow} subtype of $T$ if $S$ is a subtype of $T$ and $T$ implements only the operations $S$ defines.
Based on this, we define the notion of adjusted object as follows.

\begin{definition}
  \labdef{adjusted}
  Let $O$ and $O'$ be shared objects.
  $O$ \emph{adjusts} $O'$ when $O'.T$ is a narrow subtype of $O.T$ and $O.m \subseteq O'.m$.
\end{definition}

\paragraph{Examples}
% An object is adjusted when its access is restricted, or its interface altered.
To illustrate this principle, consider a reference object, as commonly found in programming languages.
Its state consists of one variable $s$, initially set to null ($\bot$).
The interface offers two operations:
\begin{inparaenumorig}[]
\item $\get$ returns the value of $s$, and
\item $\set(x)$ changes $s$ to $x$, provided $x$ is an address.
\end{inparaenumorig}
We can adjust this object by requiring that it is write-once.
In which case, we add the precondition $v=\bot$ to $\set$, as detailed in \reftab{spec}.
This table also presents several adjusted versions of common data types.
For each data type, its operations are specified in Hoare logic \cite{hoare}.
The reference object corresponds to $R_1$, while its adjusted counterpart is $R_2$ (fourth row in \reftab{spec}).

Another approach to adjust the reference is to have at most one thread that call $\set$ (single writer).
In this case, the access permission map of the object is changed appropriately.

\paragraph{Benefits}
When $O$ adjusts $O'$, $O'$ is at least as strong as $O$.
This is a consequence of subtyping
Formally,

\begin{restatable}{prop}{computability}
  \labprop{computability}
  Consider that object $O$ adjusts object $O'$.
  Every distributed task that is solvable with $O$ (and registers) is also solvable with $O'$.
\end{restatable}

\iflong
The proof is in \refappendix{proofs:adjusted}.
\fi
As expected, the converse proposition does not hold.
We illustrate this with \reftab{spec}.
In this table, $S_2$ adjusts $S_1$ by nullifying the return values of $\add$ and $\remove$.
Using \reftheo{distcn}, one can observe that $S_2$ is in $\cn_1$.
On the contrary, it is easy to see that the write operations of $S_1$ both have consensus power.

From what precedes, adjusting an object may reduce the amount of conflicts necessary to implement it.
As a consequence, this should \emph{boost scalability}.
This key property is established below.
It shows that adjusting an object increases the number of edges in the indistinguishability graphs.
\iflong
The proof is also in \refappendix{proofs:adjusted}.
\fi

\begin{restatable}{prop}{scalability}
  \labprop{scalability}
  Assume that $O$ adjusts $O'$.
  Let $s$ be a common state and $B$ be a bag of operations such that $B$ complies with both $O.m$ and $O'.m$.
  Then, $G_{O'.T}(B,s) \subseteq G_{O.T}(B,s)$.
\end{restatable}

Adjusting an object open opportunities for better scalability.
In the next section, we present a methodology to achieve this.
Later, we use this methodology in the construction of \sys, an efficient library of adjusted objects for Java.

\subsection{Methodology}
\labsection{adjusted:howto}

Adjusting an object involves either altering its sequential specification, or restricting access to it.
\reffigure{tree} depicts examples of both techniques.
In this figure, a shared object is a pair $(X,Y)$, where $X$ is a data type from \reftab{spec} and $Y$ the permission map.
\ALL stands for the default permissions, when threads access the full interface.
Below, we explain the various adjustments in this figure.

% delete
Deleting an operation ({\color{del_op_color}{$\xrightarrow{\text{d}}$}}) adjusts an object.
For instance, this happens with $\reset$ for the counter object.
In $C_2$, the precondition is set to $\false$.
We obtain the same result when the postcondition is voided.
This is illustrated with the read-modify-write operation ($\rmw$).
In both cases, the operation is labeling in the indistinguishability graph as it fails silently.

% stronger pre
Another approach is a stronger precondition ({\color{SPC_color}{$\xrightarrow{\text{p}}$}}).
As noted earlier, this occurs with $R_2$ which is write once in \reftab{spec}.
The rationale is the same as for a deletion:
Suppose that $d$ has a stronger precondition than $c$ and $c$'s precondition is satisfied in a permutation $x$, while $d$ is not.
Then, $d$ fails silently in $x$ which increases its likeliness to be labeling.

% weaker post
A third option is to have a weaker postcondition.
For instance, voiding the return value of a write makes it blind.
In \reftab{spec}, this happens with operation $\add$ in $S_2$ ({\color{no_ret_val_color}{$\xrightarrow{\text{r}}$}}).
This alteration densifies the indistinguishability graph, increasing the moverness of other operations (see~\refsection{dist:scalability}).

% commuting side-effects
One can also restrict how an object is accessed.
For instance, writes can be commuting ({\color{MCW_color}{$\xrightarrow{\text{c}}$}}).
In \reffigure{tree}, $(S_{3}, \CWMR)$, requires that all write operations, namely $\add$ and $\remove$, commute when executed by distinct threads.
Consequently, permutations in the indistinguishability graph that only differ in the way they are ordering these operations are connected. % FIXME

% asymetric usage
Operations may also be asymmetric, that is some threads are allowed to only call certain operations ({\color{MWSR_color}{$\xrightarrow{\text{m}}$}}).
Objects can be single-writer multi-reader (\SWMR), or multi-writer single-reader (\MWSR).
Writers may also only issue commuting operations (\CWMR and \CWSR) on the object.
All of this densifies the object's indistinguishability graphs.

As seen in \reffigure{tree}, several adjustments can be combined.
In the most general case, they form an acyclic directed graph.

\input{figures-tree}

%% file: figures-tree.tex
\begin{figure}
  \centering
  \begin{tikzpicture}[node distance={12mm}, font=\scriptsize, text opacity=1, fill opacity=0.8, main/.style = {draw=white, rectangle}, inner sep=1pt, minimum size=1.5em]
    \node[main] (1) {$(R_{1}, \ALL)$};
    \node[main] (2) [right = 0.5cm of 1, yshift = .5cm] {$(R_{2}, \ALL)$};
    \node[main] (3) [right = 0.5cm of 2, yshift = -.5cm] {$(R_{2}, \SWMR)$};
    \node[main] (4) [below = 0.5cm of 2] {$(R_{1}, \SWMR)$};
    \draw[SPC_color, thick, ->] (1) -- (2)node[midway, above, yshift=-.3em] {p};
    \draw[MWSR_color, thick, ->] (2) -- (3)node[midway, above, yshift=-.3em] {m};
    \draw[MWSR_color, thick, ->] (1) -- (4)node[midway, above, yshift=-.3em] {m};
    \draw[SPC_color, thick, ->] (4) -- (3)node[midway, above, yshift=-.3em] {p};

    \node[main] (1a) [below = .8cm of 1] {$(S_{1}, ALL)$};
    \node[main] (2a) [right = 0.5cm of 1a] {$(S_{2}, \ALL)$};
    \node[main] (3a) [right = 0.5cm of 2a] {$(S_{3}, \ALL)$};
    \node[main] (4a) [right = 0.5cm of 3a] {$(S_{3}, \CWMR)$};
    \node[main] (5a) [right = 0.5cm of 4a] {$(S_{3}, \CWSR)$};
    \draw[no_ret_val_color, thick, ->] (1a) -- (2a) node[midway, above, yshift=-.3em] {r};
    \draw[del_op_color, thick, ->] (2a) -- (3a) node[midway, above, yshift=-.3em] {d};
    \draw[MCW_color, thick, ->] (3a) -- (4a) node[midway, above, yshift=-.3em] {c};
    \draw[MWSR_color, thick, ->] (4a) -- (5a) node[midway, above, yshift=-.3em] {m};

    \node[main] (1b) [below = .7cm of 1a, xshift=-.5cm] [right] {$(C_{1}, ALL)$};
    \node[main] (2b) [right = 0.5cm of 1b] {$(C_{2}, \ALL)$};
    \node[main] (3b) [right = 0.5cm of 2b] {$(C_{3}, \ALL)$};
    \node[main] (4b) [right = 0.5cm of 3b] {$(C_{3}, \CWSR)$};
    \draw[del_op_color, thick, ->] (1b) -- (2b) node[midway, above, yshift=-.3em] {d};
    \draw[no_ret_val_color, thick, ->] (2b) -- (3b) node[midway, above, yshift=-.3em] {r};
    \draw[MWSR_color, thick, ->] (3b) -- (4b) node[midway, above, yshift=-.3em] {m};
  \end{tikzpicture} 
  \caption{
    \labfigure{tree}
    We combine subtyping (\adjustP, \adjustR) and restrictions of access (\adjustD, \adjustC, \adjustM) to adjust an object.
    \iflong\else\vspace{-1em}\fi%
  }
\end{figure}
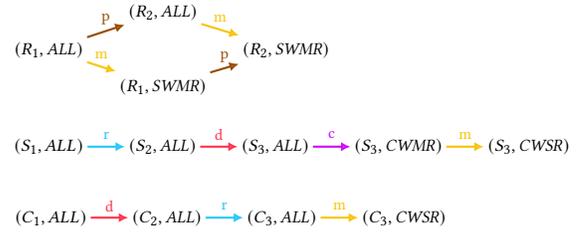

%% file: design.tex
\section{The \sys Library}
\labsection{design}

This section presents a library of adjusted objects for the Java language called \sys.

\subsection{Overview}
\labsection{design:overview}

The \sys library provides drop-in replacements for the package \code{java.util.concurrent} of the JDK---hereafter, abbreviated in \JUC.
It contains $\sim$\sloc{11,000} \cite{dego}.
\sys includes collections (e.g., set, list) and maps, as well as other base objects (e.g., reference).
Access permissions are the ones listed earlier: \ALL, \SWMR, \MWSR, \CWMR and \CWSR.
%
%% The objects listed in \reftab{spec} are all implemented in the \sys library.
%% We will detail them shortly.
For an object, \sys does not implement all the access permissions.
Instead, it provides hooks and useful abstractions to this end.
One such abstraction is a segmentation that we cover below.

\subsection{Segmentation}
\labsection{design:segmentation}

Adjusted objects that are multi-writer collections are built with \emph{segmentations}.
A segmentation is a (dynamic or static) array of objects.
Each object in the array is \SWMR and own by one thread;
it is called a \emph{segment}.
A segmentation allows to efficiently implement an adjusted object for which writes are commuting (that is, either \CWMR or \CWSR).
For instance, consider the adjusted object $(C_3,\CWSR)$ at the bottom of \reffigure{tree}.
In this case, each segment stores a counter.
When it executes $\inc$, a thread updates its segment appropriately.
To read the counter, a thread sums up all entries.
(Notice that if $\inc$ are unitary, such a read is linearizable.)

Different forms of segmentation are available in \sys.
In the case of a \code{BaseSegmentation}, the mapping between threads and segments is static.
It is implemented using an instance of the class \code{CopyOnWriteArray} and a \code{ThreadLocal} variable.
Consequently, to execute a read, e.g., when iterating over the collection, the thread needs to traverse all segments.
This makes the \code{BaseSegmentation} interesting in workloads where the object is predominantly accessed through writing.

Two other types of segmentation also exist: \code{HashSegmentation} and \code{ExtendedSegmentation}.
With these types, each data item added to a segmentation carries an information to identify the segment where it is stored.
In \code{HashSegmentation}, an item is stored in the segment corresponding to its hash code.
In the case of \code{ExtendedSegmentation}, when an item is inserted for the first time, it retains the segment where it was stored.
This is implemented using a dedicated field in the item.
These two segmentations eliminate the need to iterate over all segments for a lookup.

The next section explains how each segment is implemented.
It also covers other adjusted objects in the library.

\subsection{Implementation details}
\labsection{design:details}

To construct adjusted objects, we rely on the \code{VarHandle} mechanics. % such as the ones in a segmentation
\code{VarHandles} are introduced in JEP~193.
They permit to manipulate a Java variable while controlling its consistency at fine grain.
There are four models: \code{Plain}, \code{Opaque}, \code{Release/Acquire}, and \code{Volatile}, ordered from the weakest to the strongest.

\code{Plain} offers minimal guarantees as operations may be re-ordered. % FIXME what are they?
This model corresponds to the vanilla Java memory model.
\code{Opaque} ensures that operations on a single variable form a partial order, with writes totally ordered, thereby guaranteeing eventual consistency.
It also guarantees that the bits from multiple writes do not interleave (for writes larger than a single memory word).
The \code{Release/Acquire} model ensures causality.
For instance, if a thread $p$ changes a variable $a$ then raises a flag using \code{setRelease}, a thread $q$ must see what happened to $a$ if it sees the flag raised with \code{getAcquire},
This model is typically used in situations where only one thread can write to a variable, while others read it.
Operation in \code{Volatile} mode are linearizable.
In \sys, \code{VarHandles} are used in various \SWMR constructions.
Below, we detail a hash table and a skip list.

To construct an \SWMR object, we start from a sequential code taken, e.g., from the JDK.
This code is then extended to support concurrent readers.
For instance, when a new entry is added to an instance of \code{SWMRHashMap}, if its key is already there, it is updated with \code{setVolatile}.
Otherwise, a new node is created and inserted atomically in the appropriate bin.
Upon a call to \code{resize}, nodes cannot be re-ordered on the fly due to potential readers.
Instead, they are de-duplicated and inserted into the new binned array backing the hash table.

For \code{SWMRSkipListMap}, a new node \code{n} is added as follows:
At each level, \code{n.next} is set to \code{m}, where \code{m} is the smallest higher node than \code{n}.
Then, the \code{next} pointer of the node preceding \code{m} is changed to \code{n} using \code{setRelease}.
The last level uses \code{setVolatile} to ensure that the insertion is globally visible.

Our library also includes a multi-producer single-consumer queue.
This queue is implemented without compare-and-swap when invoking $\poll$.
Instead, the thread moves the head of the queue appropriately.
The $\offer$ operation is identical to the one in the JDK (that is, in \code{ConcurrentLinkedQueue}).

Write-once shared objects are common in applications.
For references, we use the Concurrentli implementation (\reflst{linkedin}).
For other objects, \sys uses a RCU-like mechanism, using a full copy of the object and swapping the reference atomically with \code{setVolatile}.

Notice that in \sys, read operations over adjusted objects are as consistent as in JUC.
In particular, iterating over a collection provides the same weak guarantees.

\begin{figure}[t]
  \begin{minipage}[b]{0.3\textwidth}
    \begin{tabular}{l}
      \input{figures-mining-evolution}\\
      \input{figures-mining-hotfiles}
    \end{tabular}
  \end{minipage}
  \caption{
    \labfigure{evolution}
    Number of declarations of shared objects in the ASF projects:
    \emph{(top)} on average over time, and
    \emph{(bot)} in the 20 most modified files for the latest version of each project.
  }
  \vspace{-1em}
\end{figure}
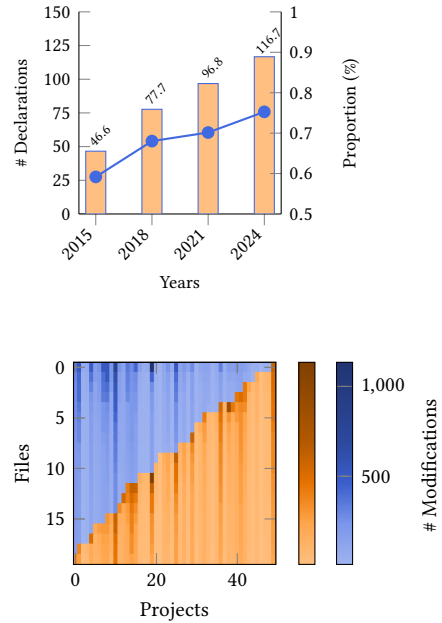

%% file: figures-mining-evolution.tex
\begin{tikzpicture}
  \begin{axis}[
      ybar,
      bar width=0.05\linewidth,
      width=0.8\textwidth,
      height=0.8\textwidth,
      xlabel={Years},
      x label style={anchor=north, below=-3mm},
      ylabel={\#~Declarations},
      y label style={
        font=\footnotesize
      },
      ymin=0, 
      ymax=150, 
      ytick={0, 25, 50, 75, 100, 125, 150},
      y tick label style={
        font=\footnotesize
      },
      xtick={2015,2018,2021,2024},
      legend style={at={(0.5,-0.25)}, anchor=north, legend columns=2},
      x label style={
        anchor=south, 
        yshift=-2em,
        font=\footnotesize
      },
      axis y line*=left, % Set left y-axis
      axis x line*=bottom,
      x tick label style={
        /pgf/number format/1000 sep=, rotate=45, 
        anchor=east,
        font=\footnotesize
      },
      every node near coord/.style={/pgf/number format/1000 sep=, rotate=45, yshift=-0.2em, xshift=0.6em, font=\tiny},
      nodes near coords
    ]

    \addplot[ 
      fill=orange!50!white, 
      draw=RoyalBlue, 
    ] coordinates {
      (2015, 46.6) (2018, 77.7)
      (2021, 96.8) (2024, 116.7)
    };

  \end{axis}

  \begin{axis}[
      width=0.8\textwidth,
      height=0.8\textwidth,
      ymin=0.5, 
      ymax=1, 
      ytick={0.5, 0.6, 0.7, 0.8, 0.9, 1},
      y tick label style={
        font=\footnotesize
      },
      axis y line*=right, % Set right y-axis
      ylabel near ticks, 
      yticklabel pos=right,
      ylabel={Proportion (\%)},
      axis x line=none,
      y label style={anchor=south, yshift=-2em,
        font=\footnotesize
      },
      legend style={at={(0,-0.35)}, anchor=north, legend columns=2},
    ]

    \addplot[RoyalBlue, thick, mark=*] coordinates {
      (2015, 0.5917169965) (2018, 0.6802743046) (2021, 0.7015159749) (2024, 0.7526356237)
    };

  \end{axis}
\end{tikzpicture}\vspace{2em}

%% file: figures-mining-hotfiles.tex
\begin{tikzpicture}
\begin{axis}[
    enlargelimits=false, 
    colorbar,  
    colormap name=CustomBlue,
    colorbar shift/.style={xshift=0.15\textwidth},
    colorbar style={
        width=0.04\textwidth,
        title={\#~Modifications},     % Titre de la colorbar
        title style={font=\small,
            yshift=-5em,
            xshift=0.25\textwidth,
            rotate=90}, % Style du titre
    },
    width=0.8\textwidth,
    height=0.8\textwidth,
    xlabel=Projects,
    ylabel=Files,
    ylabel style={font=\small},
    xlabel style={font=\small},
    tick label style={font=\small},
    ]
    \addplot [matrix plot,
        % nodes near coords=\coordindex,
        % mark=*,
        point meta=explicit,
    ] coordinates {
        
(0,0) [385] (1,0) [501] (2,0) [73] (3,0) [83] (4,0) [468] (5,0) [180] (6,0) [241] (7,0) [515] (8,0) [525] (9,0) [226] (10,0) [797] (11,0) [231] (12,0) [127] (13,0) [350] (14,0) [269] (15,0) [425] (16,0) [142] (17,0) [59] (18,0) [80] (19,0) [1134] (20,0) [47] (21,0) [22] (22,0) [107] (23,0) [233] (24,0) [89] (25,0) [522] (26,0) [135] (27,0) [198] (28,0) [126] (29,0) [268] (30,0) [25] (31,0) [106] (32,0) [128] (33,0) [126] (34,0) [35] (35,0) [77] (36,0) [318] (37,0) [56] (38,0) [180] (39,0) [153] (40,0) [211] (41,0) [165] (42,0) [140] (43,0) [88] (44,0) [34] (45,0) [50] (46,0) [9] (47,0) [42] (48,0) [38] (49,0) [nan] 

(0,1) [292] (1,1) [422] (2,1) [48] (3,1) [66] (4,1) [374] (5,1) [171] (6,1) [196] (7,1) [390] (8,1) [525] (9,1) [219] (10,1) [615] (11,1) [158] (12,1) [104] (13,1) [322] (14,1) [244] (15,1) [271] (16,1) [64] (17,1) [41] (18,1) [38] (19,1) [430] (20,1) [40] (21,1) [16] (22,1) [95] (23,1) [157] (24,1) [46] (25,1) [401] (26,1) [97] (27,1) [153] (28,1) [38] (29,1) [159] (30,1) [25] (31,1) [77] (32,1) [114] (33,1) [103] (34,1) [31] (35,1) [64] (36,1) [161] (37,1) [47] (38,1) [122] (39,1) [106] (40,1) [122] (41,1) [144] (42,1) [81] (43,1) [73] (44,1) [18] (45,1) [nan] (46,1) [nan] (47,1) [nan] (48,1) [nan] (49,1) [nan] 

(0,2) [270] (1,2) [412] (2,2) [43] (3,2) [62] (4,2) [303] (5,2) [119] (6,2) [189] (7,2) [371] (8,2) [402] (9,2) [182] (10,2) [548] (11,2) [151] (12,2) [101] (13,2) [309] (14,2) [239] (15,2) [256] (16,2) [62] (17,2) [36] (18,2) [37] (19,2) [264] (20,2) [39] (21,2) [13] (22,2) [92] (23,2) [145] (24,2) [28] (25,2) [349] (26,2) [92] (27,2) [131] (28,2) [28] (29,2) [159] (30,2) [24] (31,2) [56] (32,2) [96] (33,2) [89] (34,2) [27] (35,2) [58] (36,2) [149] (37,2) [39] (38,2) [84] (39,2) [72] (40,2) [108] (41,2) [127] (42,2) [nan] (43,2) [nan] (44,2) [nan] (45,2) [nan] (46,2) [nan] (47,2) [nan] (48,2) [nan] (49,2) [nan] 

(0,3) [265] (1,3) [411] (2,3) [41] (3,3) [51] (4,3) [230] (5,3) [116] (6,3) [175] (7,3) [341] (8,3) [355] (9,3) [115] (10,3) [506] (11,3) [140] (12,3) [101] (13,3) [278] (14,3) [220] (15,3) [230] (16,3) [53] (17,3) [32] (18,3) [36] (19,3) [191] (20,3) [37] (21,3) [12] (22,3) [69] (23,3) [127] (24,3) [27] (25,3) [327] (26,3) [59] (27,3) [100] (28,3) [26] (29,3) [154] (30,3) [23] (31,3) [51] (32,3) [93] (33,3) [83] (34,3) [26] (35,3) [49] (36,3) [130] (37,3) [27] (38,3) [81] (39,3) [63] (40,3) [nan] (41,3) [nan] (42,3) [nan] (43,3) [nan] (44,3) [nan] (45,3) [nan] (46,3) [nan] (47,3) [nan] (48,3) [nan] (49,3) [nan] 

(0,4) [261] (1,4) [349] (2,4) [38] (3,4) [48] (4,4) [205] (5,4) [114] (6,4) [143] (7,4) [234] (8,4) [261] (9,4) [109] (10,4) [476] (11,4) [108] (12,4) [99] (13,4) [171] (14,4) [215] (15,4) [215] (16,4) [40] (17,4) [31] (18,4) [32] (19,4) [187] (20,4) [35] (21,4) [11] (22,4) [69] (23,4) [110] (24,4) [27] (25,4) [260] (26,4) [51] (27,4) [96] (28,4) [25] (29,4) [125] (30,4) [20] (31,4) [50] (32,4) [89] (33,4) [44] (34,4) [25] (35,4) [39] (36,4) [nan] (37,4) [nan] (38,4) [nan] (39,4) [nan] (40,4) [nan] (41,4) [nan] (42,4) [nan] (43,4) [nan] (44,4) [nan] (45,4) [nan] (46,4) [nan] (47,4) [nan] (48,4) [nan] (49,4) [nan] 

(0,5) [259] (1,5) [316] (2,5) [36] (3,5) [47] (4,5) [169] (5,5) [105] (6,5) [125] (7,5) [234] (8,5) [245] (9,5) [98] (10,5) [429] (11,5) [106] (12,5) [95] (13,5) [163] (14,5) [210] (15,5) [179] (16,5) [38] (17,5) [28] (18,5) [32] (19,5) [162] (20,5) [32] (21,5) [11] (22,5) [67] (23,5) [109] (24,5) [22] (25,5) [256] (26,5) [47] (27,5) [95] (28,5) [24] (29,5) [121] (30,5) [18] (31,5) [49] (32,5) [nan] (33,5) [nan] (34,5) [nan] (35,5) [nan] (36,5) [nan] (37,5) [nan] (38,5) [nan] (39,5) [nan] (40,5) [nan] (41,5) [nan] (42,5) [nan] (43,5) [nan] (44,5) [nan] (45,5) [nan] (46,5) [nan] (47,5) [nan] (48,5) [nan] (49,5) [nan] 

(0,6) [246] (1,6) [290] (2,6) [36] (3,6) [45] (4,6) [169] (5,6) [93] (6,6) [109] (7,6) [226] (8,6) [233] (9,6) [80] (10,6) [375] (11,6) [97] (12,6) [94] (13,6) [159] (14,6) [210] (15,6) [160] (16,6) [38] (17,6) [27] (18,6) [29] (19,6) [159] (20,6) [21] (21,6) [11] (22,6) [61] (23,6) [107] (24,6) [22] (25,6) [241] (26,6) [40] (27,6) [91] (28,6) [24] (29,6) [120] (30,6) [nan] (31,6) [nan] (32,6) [nan] (33,6) [nan] (34,6) [nan] (35,6) [nan] (36,6) [nan] (37,6) [nan] (38,6) [nan] (39,6) [nan] (40,6) [nan] (41,6) [nan] (42,6) [nan] (43,6) [nan] (44,6) [nan] (45,6) [nan] (46,6) [nan] (47,6) [nan] (48,6) [nan] (49,6) [nan] 

(0,7) [233] (1,7) [289] (2,7) [35] (3,7) [45] (4,7) [166] (5,7) [89] (6,7) [109] (7,7) [172] (8,7) [202] (9,7) [74] (10,7) [363] (11,7) [97] (12,7) [92] (13,7) [155] (14,7) [208] (15,7) [158] (16,7) [38] (17,7) [27] (18,7) [29] (19,7) [155] (20,7) [15] (21,7) [11] (22,7) [58] (23,7) [92] (24,7) [20] (25,7) [228] (26,7) [35] (27,7) [85] (28,7) [20] (29,7) [nan] (30,7) [nan] (31,7) [nan] (32,7) [nan] (33,7) [nan] (34,7) [nan] (35,7) [nan] (36,7) [nan] (37,7) [nan] (38,7) [nan] (39,7) [nan] (40,7) [nan] (41,7) [nan] (42,7) [nan] (43,7) [nan] (44,7) [nan] (45,7) [nan] (46,7) [nan] (47,7) [nan] (48,7) [nan] (49,7) [nan] 

(0,8) [210] (1,8) [235] (2,8) [32] (3,8) [42] (4,8) [155] (5,8) [82] (6,8) [106] (7,8) [141] (8,8) [197] (9,8) [72] (10,8) [340] (11,8) [82] (12,8) [90] (13,8) [148] (14,8) [184] (15,8) [146] (16,8) [32] (17,8) [26] (18,8) [27] (19,8) [150] (20,8) [15] (21,8) [10] (22,8) [52] (23,8) [67] (24,8) [19] (25,8) [168] (26,8) [nan] (27,8) [nan] (28,8) [nan] (29,8) [nan] (30,8) [nan] (31,8) [nan] (32,8) [nan] (33,8) [nan] (34,8) [nan] (35,8) [nan] (36,8) [nan] (37,8) [nan] (38,8) [nan] (39,8) [nan] (40,8) [nan] (41,8) [nan] (42,8) [nan] (43,8) [nan] (44,8) [nan] (45,8) [nan] (46,8) [nan] (47,8) [nan] (48,8) [nan] (49,8) [nan] 

(0,9) [206] (1,9) [213] (2,9) [30] (3,9) [41] (4,9) [155] (5,9) [79] (6,9) [101] (7,9) [134] (8,9) [178] (9,9) [72] (10,9) [329] (11,9) [79] (12,9) [90] (13,9) [136] (14,9) [172] (15,9) [143] (16,9) [32] (17,9) [25] (18,9) [26] (19,9) [148] (20,9) [14] (21,9) [nan] (22,9) [nan] (23,9) [nan] (24,9) [nan] (25,9) [nan] (26,9) [nan] (27,9) [nan] (28,9) [nan] (29,9) [nan] (30,9) [nan] (31,9) [nan] (32,9) [nan] (33,9) [nan] (34,9) [nan] (35,9) [nan] (36,9) [nan] (37,9) [nan] (38,9) [nan] (39,9) [nan] (40,9) [nan] (41,9) [nan] (42,9) [nan] (43,9) [nan] (44,9) [nan] (45,9) [nan] (46,9) [nan] (47,9) [nan] (48,9) [nan] (49,9) [nan] 

(0,10) [203] (1,10) [210] (2,10) [30] (3,10) [38] (4,10) [147] (5,10) [76] (6,10) [100] (7,10) [134] (8,10) [167] (9,10) [71] (10,10) [297] (11,10) [71] (12,10) [89] (13,10) [120] (14,10) [166] (15,10) [132] (16,10) [31] (17,10) [25] (18,10) [23] (19,10) [147] (20,10) [nan] (21,10) [nan] (22,10) [nan] (23,10) [nan] (24,10) [nan] (25,10) [nan] (26,10) [nan] (27,10) [nan] (28,10) [nan] (29,10) [nan] (30,10) [nan] (31,10) [nan] (32,10) [nan] (33,10) [nan] (34,10) [nan] (35,10) [nan] (36,10) [nan] (37,10) [nan] (38,10) [nan] (39,10) [nan] (40,10) [nan] (41,10) [nan] (42,10) [nan] (43,10) [nan] (44,10) [nan] (45,10) [nan] (46,10) [nan] (47,10) [nan] (48,10) [nan] (49,10) [nan] 

(0,11) [200] (1,11) [182] (2,11) [28] (3,11) [38] (4,11) [137] (5,11) [73] (6,11) [96] (7,11) [134] (8,11) [154] (9,11) [67] (10,11) [268] (11,11) [70] (12,11) [83] (13,11) [119] (14,11) [166] (15,11) [126] (16,11) [nan] (17,11) [nan] (18,11) [nan] (19,11) [nan] (20,11) [nan] (21,11) [nan] (22,11) [nan] (23,11) [nan] (24,11) [nan] (25,11) [nan] (26,11) [nan] (27,11) [nan] (28,11) [nan] (29,11) [nan] (30,11) [nan] (31,11) [nan] (32,11) [nan] (33,11) [nan] (34,11) [nan] (35,11) [nan] (36,11) [nan] (37,11) [nan] (38,11) [nan] (39,11) [nan] (40,11) [nan] (41,11) [nan] (42,11) [nan] (43,11) [nan] (44,11) [nan] (45,11) [nan] (46,11) [nan] (47,11) [nan] (48,11) [nan] (49,11) [nan] 

(0,12) [199] (1,12) [179] (2,12) [25] (3,12) [37] (4,12) [132] (5,12) [72] (6,12) [96] (7,12) [131] (8,12) [141] (9,12) [64] (10,12) [267] (11,12) [68] (12,12) [79] (13,12) [nan] (14,12) [nan] (15,12) [nan] (16,12) [nan] (17,12) [nan] (18,12) [nan] (19,12) [nan] (20,12) [nan] (21,12) [nan] (22,12) [nan] (23,12) [nan] (24,12) [nan] (25,12) [nan] (26,12) [nan] (27,12) [nan] (28,12) [nan] (29,12) [nan] (30,12) [nan] (31,12) [nan] (32,12) [nan] (33,12) [nan] (34,12) [nan] (35,12) [nan] (36,12) [nan] (37,12) [nan] (38,12) [nan] (39,12) [nan] (40,12) [nan] (41,12) [nan] (42,12) [nan] (43,12) [nan] (44,12) [nan] (45,12) [nan] (46,12) [nan] (47,12) [nan] (48,12) [nan] (49,12) [nan] 

(0,13) [197] (1,13) [178] (2,13) [25] (3,13) [36] (4,13) [130] (5,13) [70] (6,13) [93] (7,13) [125] (8,13) [134] (9,13) [62] (10,13) [267] (11,13) [64] (12,13) [nan] (13,13) [nan] (14,13) [nan] (15,13) [nan] (16,13) [nan] (17,13) [nan] (18,13) [nan] (19,13) [nan] (20,13) [nan] (21,13) [nan] (22,13) [nan] (23,13) [nan] (24,13) [nan] (25,13) [nan] (26,13) [nan] (27,13) [nan] (28,13) [nan] (29,13) [nan] (30,13) [nan] (31,13) [nan] (32,13) [nan] (33,13) [nan] (34,13) [nan] (35,13) [nan] (36,13) [nan] (37,13) [nan] (38,13) [nan] (39,13) [nan] (40,13) [nan] (41,13) [nan] (42,13) [nan] (43,13) [nan] (44,13) [nan] (45,13) [nan] (46,13) [nan] (47,13) [nan] (48,13) [nan] (49,13) [nan] 

(0,14) [196] (1,14) [176] (2,14) [24] (3,14) [36] (4,14) [119] (5,14) [69] (6,14) [90] (7,14) [124] (8,14) [132] (9,14) [61] (10,14) [261] (11,14) [nan] (12,14) [nan] (13,14) [nan] (14,14) [nan] (15,14) [nan] (16,14) [nan] (17,14) [nan] (18,14) [nan] (19,14) [nan] (20,14) [nan] (21,14) [nan] (22,14) [nan] (23,14) [nan] (24,14) [nan] (25,14) [nan] (26,14) [nan] (27,14) [nan] (28,14) [nan] (29,14) [nan] (30,14) [nan] (31,14) [nan] (32,14) [nan] (33,14) [nan] (34,14) [nan] (35,14) [nan] (36,14) [nan] (37,14) [nan] (38,14) [nan] (39,14) [nan] (40,14) [nan] (41,14) [nan] (42,14) [nan] (43,14) [nan] (44,14) [nan] (45,14) [nan] (46,14) [nan] (47,14) [nan] (48,14) [nan] (49,14) [nan] 

(0,15) [193] (1,15) [170] (2,15) [24] (3,15) [35] (4,15) [113] (5,15) [64] (6,15) [87] (7,15) [123] (8,15) [nan] (9,15) [nan] (10,15) [nan] (11,15) [nan] (12,15) [nan] (13,15) [nan] (14,15) [nan] (15,15) [nan] (16,15) [nan] (17,15) [nan] (18,15) [nan] (19,15) [nan] (20,15) [nan] (21,15) [nan] (22,15) [nan] (23,15) [nan] (24,15) [nan] (25,15) [nan] (26,15) [nan] (27,15) [nan] (28,15) [nan] (29,15) [nan] (30,15) [nan] (31,15) [nan] (32,15) [nan] (33,15) [nan] (34,15) [nan] (35,15) [nan] (36,15) [nan] (37,15) [nan] (38,15) [nan] (39,15) [nan] (40,15) [nan] (41,15) [nan] (42,15) [nan] (43,15) [nan] (44,15) [nan] (45,15) [nan] (46,15) [nan] (47,15) [nan] (48,15) [nan] (49,15) [nan] 

(0,16) [180] (1,16) [163] (2,16) [23] (3,16) [34] (4,16) [113] (5,16) [nan] (6,16) [nan] (7,16) [nan] (8,16) [nan] (9,16) [nan] (10,16) [nan] (11,16) [nan] (12,16) [nan] (13,16) [nan] (14,16) [nan] (15,16) [nan] (16,16) [nan] (17,16) [nan] (18,16) [nan] (19,16) [nan] (20,16) [nan] (21,16) [nan] (22,16) [nan] (23,16) [nan] (24,16) [nan] (25,16) [nan] (26,16) [nan] (27,16) [nan] (28,16) [nan] (29,16) [nan] (30,16) [nan] (31,16) [nan] (32,16) [nan] (33,16) [nan] (34,16) [nan] (35,16) [nan] (36,16) [nan] (37,16) [nan] (38,16) [nan] (39,16) [nan] (40,16) [nan] (41,16) [nan] (42,16) [nan] (43,16) [nan] (44,16) [nan] (45,16) [nan] (46,16) [nan] (47,16) [nan] (48,16) [nan] (49,16) [nan] 

(0,17) [171] (1,17) [159] (2,17) [22] (3,17) [33] (4,17) [nan] (5,17) [nan] (6,17) [nan] (7,17) [nan] (8,17) [nan] (9,17) [nan] (10,17) [nan] (11,17) [nan] (12,17) [nan] (13,17) [nan] (14,17) [nan] (15,17) [nan] (16,17) [nan] (17,17) [nan] (18,17) [nan] (19,17) [nan] (20,17) [nan] (21,17) [nan] (22,17) [nan] (23,17) [nan] (24,17) [nan] (25,17) [nan] (26,17) [nan] (27,17) [nan] (28,17) [nan] (29,17) [nan] (30,17) [nan] (31,17) [nan] (32,17) [nan] (33,17) [nan] (34,17) [nan] (35,17) [nan] (36,17) [nan] (37,17) [nan] (38,17) [nan] (39,17) [nan] (40,17) [nan] (41,17) [nan] (42,17) [nan] (43,17) [nan] (44,17) [nan] (45,17) [nan] (46,17) [nan] (47,17) [nan] (48,17) [nan] (49,17) [nan] 

(0,18) [162] (1,18) [nan] (2,18) [nan] (3,18) [nan] (4,18) [nan] (5,18) [nan] (6,18) [nan] (7,18) [nan] (8,18) [nan] (9,18) [nan] (10,18) [nan] (11,18) [nan] (12,18) [nan] (13,18) [nan] (14,18) [nan] (15,18) [nan] (16,18) [nan] (17,18) [nan] (18,18) [nan] (19,18) [nan] (20,18) [nan] (21,18) [nan] (22,18) [nan] (23,18) [nan] (24,18) [nan] (25,18) [nan] (26,18) [nan] (27,18) [nan] (28,18) [nan] (29,18) [nan] (30,18) [nan] (31,18) [nan] (32,18) [nan] (33,18) [nan] (34,18) [nan] (35,18) [nan] (36,18) [nan] (37,18) [nan] (38,18) [nan] (39,18) [nan] (40,18) [nan] (41,18) [nan] (42,18) [nan] (43,18) [nan] (44,18) [nan] (45,18) [nan] (46,18) [nan] (47,18) [nan] (48,18) [nan] (49,18) [nan] 

(0,19) [nan] (1,19) [nan] (2,19) [nan] (3,19) [nan] (4,19) [nan] (5,19) [nan] (6,19) [nan] (7,19) [nan] (8,19) [nan] (9,19) [nan] (10,19) [nan] (11,19) [nan] (12,19) [nan] (13,19) [nan] (14,19) [nan] (15,19) [nan] (16,19) [nan] (17,19) [nan] (18,19) [nan] (19,19) [nan] (20,19) [nan] (21,19) [nan] (22,19) [nan] (23,19) [nan] (24,19) [nan] (25,19) [nan] (26,19) [nan] (27,19) [nan] (28,19) [nan] (29,19) [nan] (30,19) [nan] (31,19) [nan] (32,19) [nan] (33,19) [nan] (34,19) [nan] (35,19) [nan] (36,19) [nan] (37,19) [nan] (38,19) [nan] (39,19) [nan] (40,19) [nan] (41,19) [nan] (42,19) [nan] (43,19) [nan] (44,19) [nan] (45,19) [nan] (46,19) [nan] (47,19) [nan] (48,19) [nan] (49,19) [nan] 
    };
\end{axis}
\begin{axis}[
    enlargelimits=false, 
    colorbar,
    colormap name=CustomOrange,
    colorbar style={width=0.04\textwidth, ytick=\empty},
    width=0.8\textwidth,
    height=0.8\textwidth,
    ylabel=\empty,
    xlabel=\empty,
    yticklabel=\empty,
    xticklabel=\empty
    ]
    \addplot [matrix plot,
        point meta=explicit,
    ] coordinates {
(0,0) [nan] (1,0) [nan] (2,0) [nan] (3,0) [nan] (4,0) [nan] (5,0) [nan] (6,0) [nan] (7,0) [nan] (8,0) [nan] (9,0) [nan] (10,0) [nan] (11,0) [nan] (12,0) [nan] (13,0) [nan] (14,0) [nan] (15,0) [nan] (16,0) [nan] (17,0) [nan] (18,0) [nan] (19,0) [nan] (20,0) [nan] (21,0) [nan] (22,0) [nan] (23,0) [nan] (24,0) [nan] (25,0) [nan] (26,0) [nan] (27,0) [nan] (28,0) [nan] (29,0) [nan] (30,0) [nan] (31,0) [nan] (32,0) [nan] (33,0) [nan] (34,0) [nan] (35,0) [nan] (36,0) [nan] (37,0) [nan] (38,0) [nan] (39,0) [nan] (40,0) [nan] (41,0) [nan] (42,0) [nan] (43,0) [nan] (44,0) [nan] (45,0) [nan] (46,0) [nan] (47,0) [nan] (48,0) [nan] (49,0) [354] 

(0,1) [nan] (1,1) [nan] (2,1) [nan] (3,1) [nan] (4,1) [nan] (5,1) [nan] (6,1) [nan] (7,1) [nan] (8,1) [nan] (9,1) [nan] (10,1) [nan] (11,1) [nan] (12,1) [nan] (13,1) [nan] (14,1) [nan] (15,1) [nan] (16,1) [nan] (17,1) [nan] (18,1) [nan] (19,1) [nan] (20,1) [nan] (21,1) [nan] (22,1) [nan] (23,1) [nan] (24,1) [nan] (25,1) [nan] (26,1) [nan] (27,1) [nan] (28,1) [nan] (29,1) [nan] (30,1) [nan] (31,1) [nan] (32,1) [nan] (33,1) [nan] (34,1) [nan] (35,1) [nan] (36,1) [nan] (37,1) [nan] (38,1) [nan] (39,1) [nan] (40,1) [nan] (41,1) [nan] (42,1) [nan] (43,1) [nan] (44,1) [nan] (45,1) [126] (46,1) [13] (47,1) [56] (48,1) [62] (49,1) [316] 

(0,2) [nan] (1,2) [nan] (2,2) [nan] (3,2) [nan] (4,2) [nan] (5,2) [nan] (6,2) [nan] (7,2) [nan] (8,2) [nan] (9,2) [nan] (10,2) [nan] (11,2) [nan] (12,2) [nan] (13,2) [nan] (14,2) [nan] (15,2) [nan] (16,2) [nan] (17,2) [nan] (18,2) [nan] (19,2) [nan] (20,2) [nan] (21,2) [nan] (22,2) [nan] (23,2) [nan] (24,2) [nan] (25,2) [nan] (26,2) [nan] (27,2) [nan] (28,2) [nan] (29,2) [nan] (30,2) [nan] (31,2) [nan] (32,2) [nan] (33,2) [nan] (34,2) [nan] (35,2) [nan] (36,2) [nan] (37,2) [nan] (38,2) [nan] (39,2) [nan] (40,2) [nan] (41,2) [nan] (42,2) [262] (43,2) [140] (44,2) [41] (45,2) [81] (46,2) [12] (47,2) [53] (48,2) [49] (49,2) [304] 

(0,3) [nan] (1,3) [nan] (2,3) [nan] (3,3) [nan] (4,3) [nan] (5,3) [nan] (6,3) [nan] (7,3) [nan] (8,3) [nan] (9,3) [nan] (10,3) [nan] (11,3) [nan] (12,3) [nan] (13,3) [nan] (14,3) [nan] (15,3) [nan] (16,3) [nan] (17,3) [nan] (18,3) [nan] (19,3) [nan] (20,3) [nan] (21,3) [nan] (22,3) [nan] (23,3) [nan] (24,3) [nan] (25,3) [nan] (26,3) [nan] (27,3) [nan] (28,3) [nan] (29,3) [nan] (30,3) [nan] (31,3) [nan] (32,3) [nan] (33,3) [nan] (34,3) [nan] (35,3) [nan] (36,3) [nan] (37,3) [nan] (38,3) [nan] (39,3) [nan] (40,3) [267] (41,3) [308] (42,3) [257] (43,3) [66] (44,3) [30] (45,3) [73] (46,3) [12] (47,3) [41] (48,3) [43] (49,3) [291] 

(0,4) [nan] (1,4) [nan] (2,4) [nan] (3,4) [nan] (4,4) [nan] (5,4) [nan] (6,4) [nan] (7,4) [nan] (8,4) [nan] (9,4) [nan] (10,4) [nan] (11,4) [nan] (12,4) [nan] (13,4) [nan] (14,4) [nan] (15,4) [nan] (16,4) [nan] (17,4) [nan] (18,4) [nan] (19,4) [nan] (20,4) [nan] (21,4) [nan] (22,4) [nan] (23,4) [nan] (24,4) [nan] (25,4) [nan] (26,4) [nan] (27,4) [nan] (28,4) [nan] (29,4) [nan] (30,4) [nan] (31,4) [nan] (32,4) [nan] (33,4) [nan] (34,4) [nan] (35,4) [nan] (36,4) [329] (37,4) [70] (38,4) [542] (39,4) [263] (40,4) [186] (41,4) [217] (42,4) [207] (43,4) [50] (44,4) [26] (45,4) [71] (46,4) [12] (47,4) [39] (48,4) [41] (49,4) [275] 

(0,5) [nan] (1,5) [nan] (2,5) [nan] (3,5) [nan] (4,5) [nan] (5,5) [nan] (6,5) [nan] (7,5) [nan] (8,5) [nan] (9,5) [nan] (10,5) [nan] (11,5) [nan] (12,5) [nan] (13,5) [nan] (14,5) [nan] (15,5) [nan] (16,5) [nan] (17,5) [nan] (18,5) [nan] (19,5) [nan] (20,5) [nan] (21,5) [nan] (22,5) [nan] (23,5) [nan] (24,5) [nan] (25,5) [nan] (26,5) [nan] (27,5) [nan] (28,5) [nan] (29,5) [nan] (30,5) [nan] (31,5) [nan] (32,5) [257] (33,5) [52] (34,5) [43] (35,5) [88] (36,5) [284] (37,5) [68] (38,5) [192] (39,5) [106] (40,5) [172] (41,5) [201] (42,5) [145] (43,5) [49] (44,5) [24] (45,5) [70] (46,5) [11] (47,5) [36] (48,5) [38] (49,5) [274] 

(0,6) [nan] (1,6) [nan] (2,6) [nan] (3,6) [nan] (4,6) [nan] (5,6) [nan] (6,6) [nan] (7,6) [nan] (8,6) [nan] (9,6) [nan] (10,6) [nan] (11,6) [nan] (12,6) [nan] (13,6) [nan] (14,6) [nan] (15,6) [nan] (16,6) [nan] (17,6) [nan] (18,6) [nan] (19,6) [nan] (20,6) [nan] (21,6) [nan] (22,6) [nan] (23,6) [nan] (24,6) [nan] (25,6) [nan] (26,6) [nan] (27,6) [nan] (28,6) [nan] (29,6) [nan] (30,6) [39] (31,6) [99] (32,6) [171] (33,6) [50] (34,6) [35] (35,6) [66] (36,6) [251] (37,6) [46] (38,6) [153] (39,6) [88] (40,6) [127] (41,6) [193] (42,6) [144] (43,6) [47] (44,6) [21] (45,6) [63] (46,6) [10] (47,6) [36] (48,6) [37] (49,6) [250] 

(0,7) [nan] (1,7) [nan] (2,7) [nan] (3,7) [nan] (4,7) [nan] (5,7) [nan] (6,7) [nan] (7,7) [nan] (8,7) [nan] (9,7) [nan] (10,7) [nan] (11,7) [nan] (12,7) [nan] (13,7) [nan] (14,7) [nan] (15,7) [nan] (16,7) [nan] (17,7) [nan] (18,7) [nan] (19,7) [nan] (20,7) [nan] (21,7) [nan] (22,7) [nan] (23,7) [nan] (24,7) [nan] (25,7) [nan] (26,7) [nan] (27,7) [nan] (28,7) [nan] (29,7) [356] (30,7) [25] (31,7) [64] (32,7) [159] (33,7) [49] (34,7) [28] (35,7) [56] (36,7) [246] (37,7) [34] (38,7) [130] (39,7) [81] (40,7) [118] (41,7) [192] (42,7) [124] (43,7) [46] (44,7) [19] (45,7) [60] (46,7) [10] (47,7) [36] (48,7) [34] (49,7) [228] 

(0,8) [nan] (1,8) [nan] (2,8) [nan] (3,8) [nan] (4,8) [nan] (5,8) [nan] (6,8) [nan] (7,8) [nan] (8,8) [nan] (9,8) [nan] (10,8) [nan] (11,8) [nan] (12,8) [nan] (13,8) [nan] (14,8) [nan] (15,8) [nan] (16,8) [nan] (17,8) [nan] (18,8) [nan] (19,8) [nan] (20,8) [nan] (21,8) [nan] (22,8) [nan] (23,8) [nan] (24,8) [nan] (25,8) [nan] (26,8) [76] (27,8) [150] (28,8) [53] (29,8) [220] (30,8) [24] (31,8) [64] (32,8) [125] (33,8) [45] (34,8) [26] (35,8) [52] (36,8) [219] (37,8) [34] (38,8) [116] (39,8) [77] (40,8) [103] (41,8) [185] (42,8) [112] (43,8) [46] (44,8) [19] (45,8) [56] (46,8) [10] (47,8) [35] (48,8) [34] (49,8) [223] 

(0,9) [nan] (1,9) [nan] (2,9) [nan] (3,9) [nan] (4,9) [nan] (5,9) [nan] (6,9) [nan] (7,9) [nan] (8,9) [nan] (9,9) [nan] (10,9) [nan] (11,9) [nan] (12,9) [nan] (13,9) [nan] (14,9) [nan] (15,9) [nan] (16,9) [nan] (17,9) [nan] (18,9) [nan] (19,9) [nan] (20,9) [nan] (21,9) [13] (22,9) [96] (23,9) [135] (24,9) [62] (25,9) [381] (26,9) [61] (27,9) [143] (28,9) [44] (29,9) [185] (30,9) [23] (31,9) [63] (32,9) [119] (33,9) [43] (34,9) [22] (35,9) [52] (36,9) [217] (37,9) [33] (38,9) [96] (39,9) [76] (40,9) [102] (41,9) [179] (42,9) [111] (43,9) [45] (44,9) [19] (45,9) [55] (46,9) [10] (47,9) [35] (48,9) [34] (49,9) [220] 

(0,10) [nan] (1,10) [nan] (2,10) [nan] (3,10) [nan] (4,10) [nan] (5,10) [nan] (6,10) [nan] (7,10) [nan] (8,10) [nan] (9,10) [nan] (10,10) [nan] (11,10) [nan] (12,10) [nan] (13,10) [nan] (14,10) [nan] (15,10) [nan] (16,10) [nan] (17,10) [nan] (18,10) [nan] (19,10) [nan] (20,10) [71] (21,10) [13] (22,10) [75] (23,10) [133] (24,10) [37] (25,10) [288] (26,10) [49] (27,10) [139] (28,10) [42] (29,10) [165] (30,10) [22] (31,10) [59] (32,10) [104] (33,10) [39] (34,10) [21] (35,10) [52] (36,10) [208] (37,10) [29] (38,10) [95] (39,10) [75] (40,10) [98] (41,10) [164] (42,10) [104] (43,10) [42] (44,10) [18] (45,10) [53] (46,10) [9] (47,10) [34] (48,10) [34] (49,10) [219] 

(0,11) [nan] (1,11) [nan] (2,11) [nan] (3,11) [nan] (4,11) [nan] (5,11) [nan] (6,11) [nan] (7,11) [nan] (8,11) [nan] (9,11) [nan] (10,11) [nan] (11,11) [nan] (12,11) [nan] (13,11) [nan] (14,11) [nan] (15,11) [nan] (16,11) [101] (17,11) [53] (18,11) [48] (19,11) [678] (20,11) [57] (21,11) [12] (22,11) [74] (23,11) [108] (24,11) [34] (25,11) [284] (26,11) [42] (27,11) [118] (28,11) [25] (29,11) [164] (30,11) [22] (31,11) [57] (32,11) [103] (33,11) [39] (34,11) [21] (35,11) [51] (36,11) [204] (37,11) [28] (38,11) [86] (39,11) [73] (40,11) [98] (41,11) [159] (42,11) [95] (43,11) [37] (44,11) [18] (45,11) [53] (46,11) [9] (47,11) [33] (48,11) [33] (49,11) [194] 

(0,12) [nan] (1,12) [nan] (2,12) [nan] (3,12) [nan] (4,12) [nan] (5,12) [nan] (6,12) [nan] (7,12) [nan] (8,12) [nan] (9,12) [nan] (10,12) [nan] (11,12) [nan] (12,12) [nan] (13,12) [252] (14,12) [370] (15,12) [359] (16,12) [73] (17,12) [42] (18,12) [44] (19,12) [368] (20,12) [47] (21,12) [12] (22,12) [74] (23,12) [108] (24,12) [32] (25,12) [254] (26,12) [35] (27,12) [98] (28,12) [25] (29,12) [147] (30,12) [21] (31,12) [56] (32,12) [101] (33,12) [39] (34,12) [20] (35,12) [48] (36,12) [191] (37,12) [28] (38,12) [82] (39,12) [70] (40,12) [93] (41,12) [157] (42,12) [94] (43,12) [36] (44,12) [18] (45,12) [50] (46,12) [9] (47,12) [33] (48,12) [33] (49,12) [191] 

(0,13) [nan] (1,13) [nan] (2,13) [nan] (3,13) [nan] (4,13) [nan] (5,13) [nan] (6,13) [nan] (7,13) [nan] (8,13) [nan] (9,13) [nan] (10,13) [nan] (11,13) [nan] (12,13) [327] (13,13) [230] (14,13) [313] (15,13) [258] (16,13) [54] (17,13) [32] (18,13) [40] (19,13) [265] (20,13) [41] (21,13) [11] (22,13) [70] (23,13) [105] (24,13) [29] (25,13) [213] (26,13) [35] (27,13) [97] (28,13) [25] (29,13) [132] (30,13) [21] (31,13) [54] (32,13) [95] (33,13) [39] (34,13) [20] (35,13) [48] (36,13) [187] (37,13) [25] (38,13) [75] (39,13) [67] (40,13) [91] (41,13) [156] (42,13) [89] (43,13) [34] (44,13) [17] (45,13) [48] (46,13) [9] (47,13) [33] (48,13) [33] (49,13) [181] 

(0,14) [nan] (1,14) [nan] (2,14) [nan] (3,14) [nan] (4,14) [nan] (5,14) [nan] (6,14) [nan] (7,14) [nan] (8,14) [nan] (9,14) [nan] (10,14) [nan] (11,14) [214] (12,14) [118] (13,14) [167] (14,14) [261] (15,14) [203] (16,14) [45] (17,14) [31] (18,14) [34] (19,14) [259] (20,14) [26] (21,14) [10] (22,14) [70] (23,14) [89] (24,14) [29] (25,14) [195] (26,14) [35] (27,14) [88] (28,14) [24] (29,14) [131] (30,14) [19] (31,14) [49] (32,14) [94] (33,14) [39] (34,14) [19] (35,14) [46] (36,14) [158] (37,14) [24] (38,14) [74] (39,14) [66] (40,14) [90] (41,14) [140] (42,14) [79] (43,14) [34] (44,14) [17] (45,14) [46] (46,14) [9] (47,14) [32] (48,14) [31] (49,14) [171] 

(0,15) [nan] (1,15) [nan] (2,15) [nan] (3,15) [nan] (4,15) [nan] (5,15) [nan] (6,15) [nan] (7,15) [nan] (8,15) [188] (9,15) [143] (10,15) [352] (11,15) [104] (12,15) [99] (13,15) [158] (14,15) [241] (15,15) [167] (16,15) [40] (17,15) [29] (18,15) [31] (19,15) [200] (20,15) [17] (21,15) [10] (22,15) [64] (23,15) [85] (24,15) [24] (25,15) [171] (26,15) [35] (27,15) [83] (28,15) [23] (29,15) [129] (30,15) [18] (31,15) [48] (32,15) [88] (33,15) [39] (34,15) [19] (35,15) [45] (36,15) [158] (37,15) [24] (38,15) [71] (39,15) [61] (40,15) [89] (41,15) [136] (42,15) [79] (43,15) [34] (44,15) [17] (45,15) [46] (46,15) [8] (47,15) [32] (48,15) [30] (49,15) [166] 

(0,16) [nan] (1,16) [nan] (2,16) [nan] (3,16) [nan] (4,16) [nan] (5,16) [124] (6,16) [160] (7,16) [175] (8,16) [177] (9,16) [113] (10,16) [349] (11,16) [90] (12,16) [92] (13,16) [124] (14,16) [240] (15,16) [139] (16,16) [40] (17,16) [29] (18,16) [30] (19,16) [192] (20,16) [16] (21,16) [10] (22,16) [63] (23,16) [74] (24,16) [22] (25,16) [167] (26,16) [34] (27,16) [82] (28,16) [23] (29,16) [128] (30,16) [18] (31,16) [48] (32,16) [87] (33,16) [39] (34,16) [18] (35,16) [44] (36,16) [153] (37,16) [24] (38,16) [67] (39,16) [58] (40,16) [86] (41,16) [133] (42,16) [73] (43,16) [34] (44,16) [16] (45,16) [45] (46,16) [8] (47,16) [32] (48,16) [30] (49,16) [160] 

(0,17) [nan] (1,17) [nan] (2,17) [nan] (3,17) [nan] (4,17) [304] (5,17) [80] (6,17) [108] (7,17) [151] (8,17) [151] (9,17) [92] (10,17) [292] (11,17) [72] (12,17) [90] (13,17) [123] (14,17) [206] (15,17) [135] (16,17) [38] (17,17) [26] (18,17) [27] (19,17) [183] (20,17) [15] (21,17) [10] (22,17) [63] (23,17) [73] (24,17) [22] (25,17) [157] (26,17) [34] (27,17) [82] (28,17) [21] (29,17) [117] (30,17) [18] (31,17) [47] (32,17) [83] (33,17) [39] (34,17) [18] (35,17) [43] (36,17) [146] (37,17) [24] (38,17) [66] (39,17) [58] (40,17) [83] (41,17) [128] (42,17) [73] (43,17) [33] (44,17) [15] (45,17) [41] (46,17) [8] (47,17) [31] (48,17) [30] (49,17) [159] 

(0,18) [nan] (1,18) [233] (2,18) [28] (3,18) [52] (4,18) [117] (5,18) [76] (6,18) [95] (7,18) [136] (8,18) [143] (9,18) [81] (10,18) [276] (11,18) [70] (12,18) [86] (13,18) [120] (14,18) [203] (15,18) [130] (16,18) [35] (17,18) [26] (18,18) [27] (19,18) [164] (20,18) [14] (21,18) [10] (22,18) [59] (23,18) [64] (24,18) [19] (25,18) [156] (26,18) [34] (27,18) [80] (28,18) [21] (29,18) [117] (30,18) [17] (31,18) [46] (32,18) [82] (33,18) [39] (34,18) [17] (35,18) [39] (36,18) [143] (37,18) [24] (38,18) [63] (39,18) [58] (40,18) [81] (41,18) [123] (42,18) [72] (43,18) [32] (44,18) [15] (45,18) [41] (46,18) [8] (47,18) [30] (48,18) [29] (49,18) [152] 

(0,19) [209] (1,19) [162] (2,19) [23] (3,19) [34] (4,19) [109] (5,19) [70] (6,19) [85] (7,19) [125] (8,19) [141] (9,19) [62] (10,19) [269] (11,19) [67] (12,19) [84] (13,19) [111] (14,19) [149] (15,19) [127] (16,19) [34] (17,19) [25] (18,19) [24] (19,19) [143] (20,19) [14] (21,19) [10] (22,19) [54] (23,19) [61] (24,19) [19] (25,19) [155] (26,19) [33] (27,19) [76] (28,19) [20] (29,19) [117] (30,19) [17] (31,19) [45] (32,19) [81] (33,19) [39] (34,19) [17] (35,19) [39] (36,19) [129] (37,19) [23] (38,19) [61] (39,19) [58] (40,19) [79] (41,19) [118] (42,19) [71] (43,19) [30] (44,19) [15] (45,19) [39] (46,19) [8] (47,19) [28] (48,19) [29] (49,19) [138] 
    };
\end{axis}
\end{tikzpicture}

%% file: evaluation.tex
\section{Evaluation}
\labsection{evaluation}

This section analyzes the use of shared objects in practice.
Then, it evaluates the adjusted objects of the \sys library using micro-benchmarks and a social network application.

\subsection{How shared objects are used?}
\labsection{evaluation:usage}

We analyze the usage of shared objects in parallel Java programs.
To this end, we mined 50 prominent projects from the \textit{Apache Software Foundation} (ASF).
Our study covers the following four data types: \code{ConcurrentLinkedQueue}, \code{ConcurrentSkipListSet}, as well as \code{AtomicLong} and \code{ConcurrentHashMap}.

First, we examine the number of declarations and their proportions within each project.
\reffigure{evolution}(top) depicts the results for \code{ConcurrentHashMap}. % FIXME legend in the figure
Our findings are similar for the other data types.
Overall, we observe a gradual increase in the number of declarations.
On average, this represents less than 1\% of the total declarations in a project (second y-axis of the figure).
Nonetheless, relative to the number of lines of code, this reflects a 25\% growth over ten years.

Although shared objects are used parsimoniously, they may still be pivotal in a code base.
To evaluate this, we analyze for each ASF project the 20 most modified files over the past ten years.
For each file, we check if it uses objects from the \texttt{java.util.concurrent} library or not.
The results are reported in \reffigure{evolution}(bot).

In \reffigure{evolution}(bot), the x-axis represents the 50 projects.
Each row corresponds to the 20 most modified files in a project.
Files using (at least) one shared object are in blue, while the others are in orange.
Color intensity reflects the frequency of modifications, with darker shades indicating more commits.
As illustrated in the figure, nearly half of the most modified files involve \texttt{java.util.concurrent} objects.
This underlines that despite their low declaration count, these objects are central to program development.

We now study the extent to which the interface is used.
For this, we analyze the methods called, the classes in which these calls occur, and whether the return values are used.

\reffigure{pie} illustrates the usages we found for the aforementioned shared objects.
For clarity, only methods representing more than $10\%$ of the total calls are reported.
The other methods are grouped into the category \textit{others}.
The number in bracket indicates how many methods are in this category.

\begin{figure}[t]
  \input{figures-mining-distribution}
  \caption{
    \labfigure{pie}    
    Most used methods in the ASF projects.
  }
  \vspace{-1em}
\end{figure}
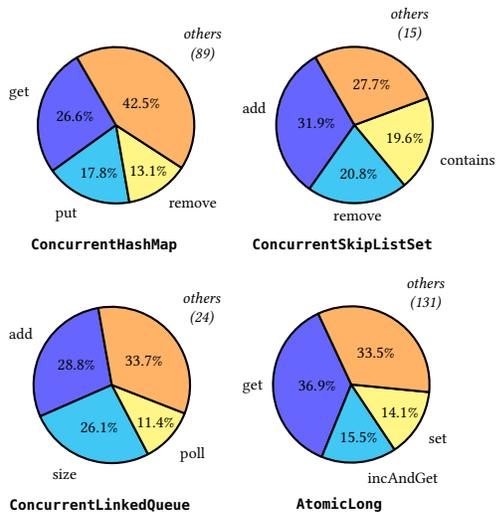

In \reffigure{pie}, we observe that only a fraction of the interface of each object is used, with some of the methods being called much more frequently than others.
In detail, out of the 134 different methods available in \code{AtomicLong}, just 3 of them account for $66.5\%$ of the calls. 
Similarly, from the 92 different methods in \code{ConcurrentHashMap}, only 3 account for $57.5\%$ of the usage.
This is equal to $66.3\%$ and $72.3\%$ for the classes \code{ConcurrentLinkedQueue} and \code{ConcurrentSkipListSet}, respectively.

\reffigure{mining}(right) depicts the usage of \code{AtomicLong} in Apache Cassandra.
The y-axis lists the methods used in this project.
For each method, the figure indicates whether the return value is used (\textcolor{red}{+}) or not (\textcolor{blue}{$\times$}).
Some methods tend to be invoked together, such as \code{get} and \code{incrementAndGet}.
For \code{AtomicLong}, the two most frequently used methods are called together in $29.5\%$ of the classes 
This value changes to $40\%$ and $45.2\%$ with \code{ConcurrentSkipListSet} and \code{ConcurrentHashMap}, respectively.
In addition, we notice that in many cases (e.g., for \code{incrementAndGet} and \code{addAndGet}), these calls do not use the return values.
Similar observations can be drawn with other programs and for the other data types.

All in all, our study reveals that there are opportunities to adjust a shared object for a specific usage in a program.

\begin{keytakeaway}[usage] 
  Shared objects are rare (<1\%) but critical, with 25\% growth in a decade.
  Each program has a limited use of an object's interface.
\end{keytakeaway}

\subsection{Micro-benchmarks}
\labsection{evaluation:micro}

We now evaluate the objects proposed in the \sys library using micro-benchmarks.

\begin{figure*}[t]
  \centering
  \hspace{-2.5em}%
  \resizebox {1.04\textwidth} {!} {
    \input{results-highcontention}
  }
  ~~~~
  \caption{
    \labfigure{perfjava:globfig}
    Performance of the objects in \sys versus their counterparts in \code{java.util.concurrent} under high contention.
  }
\end{figure*}
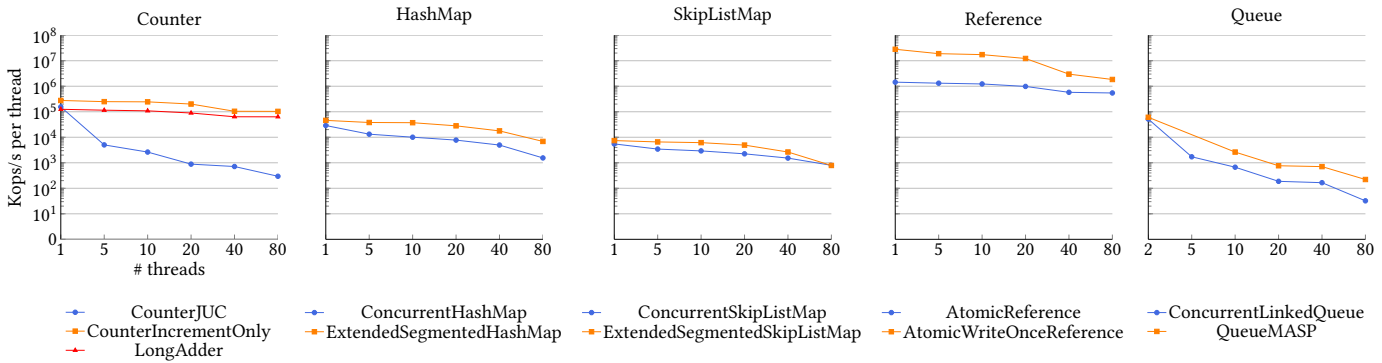

\paragraph{Settings}
The evaluation uses a machine with 362GB of DRAM and 4 Intel Xeon Gold 6230 CPUs, each having 40 (hyperthreaded) cores.
The machine runs Linux 6.1.0-18-amd64 and Java 22-oracle (OpenJDK).
Benchmarks that execute with 40 threads (or below) are run on a single socket.

For \sys, we evaluate the following objects:
\code{QueueMASP} which is a multi-producer single-consumer linked queue $(Q_1, \MWSR)$;
\code{CounterIncrementOnly}, which corresponds to $(C_3,\CWSR)$;
we implement $(M_1,\CWMR)$ using an extended segmentation called \code{ExtendedSegmentedHashMap};
the same segmentation is used in an efficient skiplist map (\code{ExtendedSegmentedSkipListMap}); and
\code{AtomicWriteOnceReference} is the class presented in \reflst{linkedin}.
These adjusted objects are compared against their counterparts in the JDK.
Such objects are found in the \code{java.util.concurrent} package, or \JUC hereafter.

% \iflong\else\vspace{-.2em}\fi%
\paragraph{Methodology}
We run each benchmark for \second{60} after a \second{30} warm-up phase.
Each reported value is an average over 30 tests.
An operation (method call) executes $1,000$ times and the measure is then averaged.
% This avoids spending too many CPU cycles in \code{System.nanoTime}.

For data structures (collections and maps), threads execute commuting updates.
This corresponds to a common pattern where each request is routed to a particular thread (using, e.g., the hash of the data item). % FIXME
A data structure starts initially with 16K data items, and stores up to 32K possible values.
Each item is generated randomly with a uniform distribution.
Snapshot-like operations (such as iterators) are not tested.
We also do not evaluate update operations that are composite (e.g., \code{putAll}).

First, we benchmark the adjusted objects in \reftab{spec} under high-contention scenarios.
% Then, we evaluate the maps with different update ratios.
%
This is equivalent to running Synchrobench \cite{synchrobench} with the following parameters:
\textit{-u100-f1-l60000-s0-a0-i[16384]-r[32768]-W30-n30}.
For each benchmark, we compute the throughput \emph{per thread}.
Since the results are per thread, an horizontal line indicates that the object scales perfectly.
Conversely, a line that goes down quickly indicates poor scalability.
If the line goes up, the object is hyper-scalable.
This may happen due to benefit of sharing the CPU cache among hardware threads.

When multiple threads compete for access to the same object, one thread may be able to obtain a lock to utilize the object, while the others wait until the lock is available.
This phenomenon can be observed by examining the \textit{cycle\_activity.stalls\_total} event in the \textit{perf} performance counter, which records the number of cycles during the execution of the program where at least one hardware thread is waiting.

We seek a correlation between throughput and the number of cycles with at least one thread waiting.
To this end, we employ the \textit{Pearson correlation coefficient}.
The coefficient varies between $-1$ and $1$ for two groups of values $A$ and $B$.
The closer the coefficient is to $1$, the more it indicates that if $A$ increases, $B$ also increases linearly.
Conversely, if the coefficient approaches $-1$, then as $A$ increases, $B$ decreases.
A coefficient close to $0$ signifies a lack of correlation.

%\iflong\else\vspace{-.2em}\fi%
\paragraph{High contention}
\reffigure{perfjava:globfig} compares \sys versus JUC in a high-contention scenario.
For maps, \code{put} is the unique operation called.
For counters, threads repeatedly call \code{incrementAndGet}.
The workload for queues is a producer-consumer:
all the threads perform \code{offer}, but one that executes just \code{poll}.
For references, threads call continuously \code{get}, once the object is initialized.

\code{CounterIncrementOnly} reduces by 80\% the number of \textit{cycle\_activity.stalls\_total} events wrt. \code{AtomicLong}.
This is the main reason for its better performance.
In this experiment, the Pearson coefficient is $-0.93$.
Hence, as the number of \textit{cycle\_activity.stalls\_total} events increases, the throughput decreases.
This phenomenon is also observed for other objects, with an average Pearson coefficient of -0.88.
When 80 threads access the object concurrently, \code{CounterIncrementOnly} is \speedup{350} faster than \code{AtomicLong}.
JUC also includes a \code{LongAdder} class which employs an approach similar to ours to alleviate contention.
Internally, \code{LongAdder} relies on \code{Striped64} which uses \code{weakCompare\&Set} for updates.
Because there is a single owner per segment, \code{CounterIncrementOnly} exclusively relies on longs.
This explains why its performance is slightly higher.

\code{ExtendedSegmentedHashMap} is up to \speedup{4.4} faster than its counterpart \code{ConcurrentHashMap}.
This improvement is explained with \textit{cycle\_activity.stalls\_total} events that are 23\% fewer for the adjusted object.
Similarly, the skip list map in \sys is up to \speedup{1.7} faster than its \JUC counterpart.
This is also explained by the fact that threads are less stalled in the implementation.

%% \code{AtomicReference} executes 37\% more instructions per cycle than \code{AtomicWriteOnceReference}.
%% Nevertheless, it is  faster (on average ).
\code{AtomicReference} employs a volatile variable which requires to use barriers.
When a read is performed, two barriers are applied after the value contained in the variable is loaded:
a \code{LoadLoad} barrier to ensure that the read is not reordered with reads performed after the barrier,
and a \code{LoadStore} barrier to ensure that writes performed after the barrier are not ordered before the reads that happened before.
A third \code{StoreLoad} barrier is used to guarantee that operations on the same volatile variable are not re-ordered.
However, since this barrier is applied subsequent to a write, it is not needed in this benchmark.

As seen in \reflst{linkedin}, the adjusted version maintains a copy of the reference.
This eliminates the need to use a barrier when reading and brings a substantial performance improvement:
\speedup{11.5} on average in \reffigure{perfjava:globfig}.

As noted in \refsection{design:details}, our adjusted queue uses a simpler mechanism to update the head when a single thread executes $\poll$ operations.
This explains why \code{QueueMASP} is an average \speedup{4.3} faster than \code{ConcurrentLinkedQueue} in \reffigure{perfjava:globfig}.

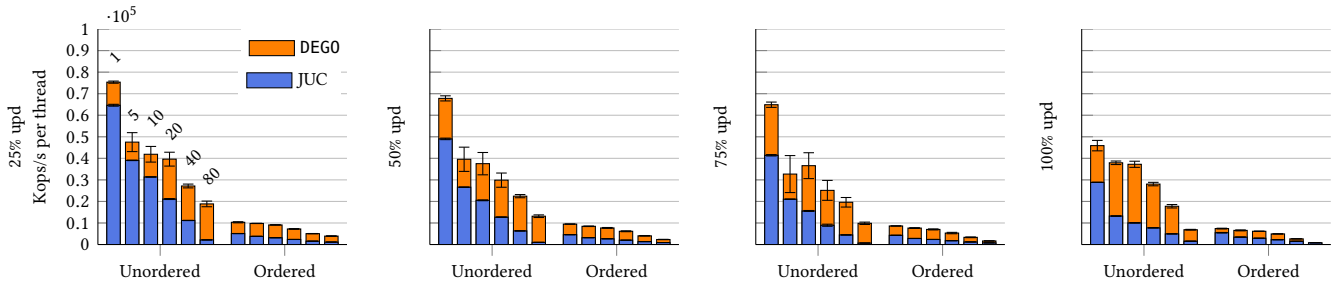
\begin{figure*}[t]
  \centering
  \input{results-mixedworkload}
  %\iflong\else\vspace{-1em}\fi%
  \caption{
    \labfigure{evaluation:mixed}
    Varying the update ratio for a hash table (Unordered) and a skip list (Ordered).
  }
  %\iflong\else\vspace{-1em}\fi%
\end{figure*}

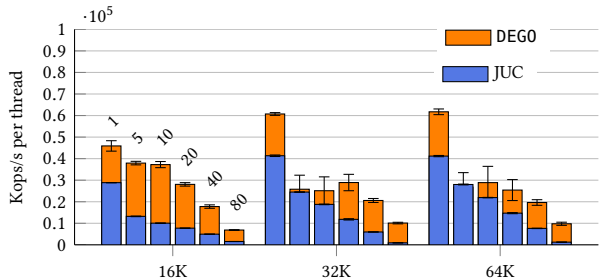
\begin{figure}
    \begin{tikzpicture}
        \input{results-sizes}
    \end{tikzpicture}
    %\vspace{-1em}%
    \caption{Performance of a hash table with various working sets (75\% of updates)}
    \labfigure{mapsizes}
    %\vspace{-1em}%
\end{figure}

\paragraph{Mixed workloads}
\reffigure{evaluation:mixed} presents the performance of two maps: a hash table and a skip list.
% As previously, the figure compares the implementations available in \sys and \code{java.util.concurrent} (\JUC).
%
Here, updates are split evenly between adds and removes.
As previously, threads updates distinct items.
Reads looks for a single item in the map.

The results in \reffigure{evaluation:mixed} indicate that for both maps, adjusted or not, the throughput tends to decline with more updates.
This is primarily due to contention as threads try to access the same memory resources on the machine.
Overall, adjusted objects perform better.
This comes from their internals which permit threads to access different memory locations, thereby reducing contention.

As pointed out earlier, \textit{cycle\_activity.stalls\_total} reports thread activity.
On average, \code{ExtendedSegmentedHashMap} exhibits 30\% less cycles with one thread waiting than \code{ConcurrentHashMap}.
With \code{ExtendedSegmentedSkipListMap}, the reduction is 11\%.

When the update ratio increases, the performance gap between \sys and \JUC widens.
For instance, with a hash table, \sys is on average \speedup{2.5} faster with 25\% updates, and up to \speedup{4.5} faster with 100\% updates.

\begin{figure*}[t]
  \centering
  \hspace{-7em}%
  \begin{minipage}[b]{0.6\textwidth}
    \centering
    \input{results-retwis-results}
  \end{minipage}
  \hspace{5em}%
  \begin{minipage}[b]{0.2\textwidth}
    \centering
    \input{results-retwis-alpha}
  \end{minipage}
\end{figure*}

\paragraph{Impact of the working set}
\reffigure{mapsizes} evaluates the hash tables when the working set changes.
On the left of the figure, the working set is the same as in prior experiments.
A map contains initially 16K data items, and there are 32K items in total.
We double these numbers twice to see how performance changes (right of \reffigure{mapsizes}).

Increasing the working set results in a growth of the binned array backing the hash table.
Each time the size doubles, the values in the bins are re-balanced, to ensure fast access to the objects.
As a consequence, contention on a given bin reduces.
Contention also decreases due to a larger working set:
the likeliness of two threads to have concurrently the same hash is smaller.
All of this explains why the performance gap between the two implementations narrows with larger working sets.

\begin{keytakeaway}[micro-benchmarks]
  \sys clearly outperforms JUC in high-contention scenarios and with mixed workloads.
  This is due to reduced cycle stalls and more parallel memory accesses.
  As the working set grows, contention decreases, narrowing the performance gap.
\end{keytakeaway}  

\subsection{A social network application}
\labsection{evaluation:retwis}

This section evaluates the interest of \sys for a modern application.
More specifically, we consider a Retwis-like benchmark, a simplified clone of Twitter.
Our implementation is multithreaded and in Java.
It covers around \sloc{1,500}.

\paragraph{Overview}
Our code extends Retwis' logic \cite{retwis}.
The social network application maintains a set of users.
Each user can write a message, follow/unfollow other users, and display her timeline (that is, the messages published by people she follows).
Additionally, a user can also join/leave a group of interest, and update her profile.
The application uses five data structures:
\begin{inparaenumorig}[]%
\item \code{mapFollowers} associates to each user its followers;
\item \code{mapFollowing} stores the users each user follows; and
\item \code{mapTimelines} maintains the timeline of each user.
\item \code{mapProfiles} maintains the profile of each user.
\item \code{community} stores the users who joined the group of interest.
\end{inparaenumorig}

When a user is added to the system, the above data structures are updated appropriately.
%% A set to store the followers and the users that the user follows.
%% These sets are updated during follow/unfollow operations.
For the timeline, a queue is used.
When a user tweets a message, the message is added to the timelines of her followers.
To limit performance impact, the message is sent only to the first followers.
Insertion in the other timelines is done asynchronously (not implemented).
When reading the timeline, all the messages in the queue are fetched, and the last 50 messages are returned.
%
% All the above data structures are accessed concurrently.
Each thread is assigned a range of users with consistent hashing.
The number of threads in the application is tunable.

We compare three versions of the application.
The first version is based on \JUC.
The second one is disjoint-access parallel (DAP).
This means that threads always access distinct objects.
DAP provides an upper bound on the performance of a parallel implementation.
The third version uses \sys.
Objects are adjusted as follow:
\code{mapFollowers}, \code{mapFollowing} as well as \code{mapTimelines} are all \CWMR.
The queue used in the timeline is multi-producer single-consumer.
The set implementing the interest group is \CWMR.
% The sets that store the the followers and the people following a given user are respectively \CWMR and \CWCR.

\paragraph{Dataset}
To build the social network graph, we follow the method in \cite{SocialNetworkGraph}.
This results in a graph whose in-degree and out-degree distributions abide by a power law;
as found in real-world social networks, e.g., Twitter \cite{GraphTwitter}.
In \cite{SocialNetworkGraph}, the authors include an extra step to increase the average coefficient cluster of the graph.
This steps is omitted in our benchmark because it is too much time consuming at the scales we consider (for $10^6$ users, it would run for hundreds of days).

\paragraph{The benchmark}
The workload is modeled after actual usage in real-world social networks.
\reftab{WorkloadRetwis} details the mix of operations.
The benchmark runs 10 times during \second{20} after a warm-up period of \second{5}.
The benchmark selects in advance users that do operations and the users they follow/unfollow.
These two distributions are modeled after a power law to capture that some users are more active (or popular) than others.
The distribution is tuned with a parameter $\alpha$.
When $\alpha$ equals $1$, it is biased and when it is close to $0$ the distribution is uniform.

As the benchmark progresses, the social network changes.
To maintain its invariants, when a user performs follow/unfollow, it immediately applies the converse operation.
This second operation is not measured in the benchmark.

\vspace{1em}%
\begin{table}[!h]
  \centering
  \small
  \begin{tabular}{|P{3.5cm}|P{2cm}|}
    \hline Add a user             & $5 \%$     \\
    \hline Follow/unfollow a user   & $5 \%$    \\
    \hline Post a tweet           & $15 \%$    \\
    \hline Display the timeline        & $60 \%$    \\
    \hline Join/leave the interest group        & $5 \%$    \\
    \hline Update the profile        & $10 \%$    \\
    \hline
  \end{tabular}
  \caption{
    \labtab{WorkloadRetwis}
    Workload for the social network application.
  }
  \vspace{-2em}
\end{table}

\paragraph{Results}
\reffigure{retwis} shows the performance of the social network application.
In this experiment, the parameter $\alpha$ is set to $1$.
We vary the number of users and the number of threads.
Performance is reported relative to \JUC.
In this figure, \sys is consistently faster than the baseline, except when a single thread runs the benchmark.
More precisely, \sys is between \speedup{0.89} and \speedup{1.7} of the baseline.
The largest speed-up is obtained with $10^5$ users, when using $80$ threads.
This is close to the performance obtained with the DAP implementation (right of \reffigure{retwis}).

Many factors influence performance, ranging from the CPU cache to the size of the working set.
For \sys, write amplification due to additional metadata can lower benefits.
In an early version of the implementation, the sets that store the followers and the people following a given user were both adjusted (because these are \CWSR).
However, the impact on memory was hindering any benefit from lower contention.

\reffigure{alpha} illustrates how the distribution law to pick users impacts performance.
When the law is biased, high locality favors \sys, as contention is the dominating performance factor.
Conversely, when users are picked uniformly at random, the CPU cache is less efficient.
Hence the performance gap with \JUC is smaller.
Contrarily to DAP, \sys and JUC share data among threads.
In this experiment, \sys has a lower cache miss ratio than DAP's---around 3 percentage points.
This explains why it is slightly better.

\begin{keytakeaway}[Retwis Benchmark with \sys]  
  \sys outperforms JUC, achieving up to \speedup{1.7} speed-up with 80 threads.
  Its advantage comes from lower contention, although write amplification may reduce gains.
  Performance also depends on the workload distribution and its locality.
\end{keytakeaway}  

% This shows the trade-off that exists between better contention and lower memory footprint when designing shared objects.

%% \subsection{Apache Ignite}
%% \labsection{evaluation:ignite}

%% We then analyse an open source database to see whether adjusted objects are used and, if so, how they are used and what benefits they provide.
%% We therefore used YCSB, a well-known tool that provides a series of benchmarks for evaluating the performance of different "key-value" and "cloud" serving stores. 
%% We choose to evaluate the open-source project Apache Ignite, a distributed database for high‑performance applications with in‑memory speed.

%% We first analyse the Apache Ignite code to see if there were any shared objects similar to our adjusted objects, and we find that a FastSizeDeqeue object is used.
%% This object keeps track of the size of the queue using a counter, which avoids having to read the whole queue to find out its size, thus improving performance.
%% We execute all YCSB core workloads using the following settings: 100 million operations are carried out, there are 1 million records, 190 clients are involved and there is no data persistence.
%% The results depicted in Figure \ref{fig:Ignite_YCSB} illustrate that, across all workloads, the utilization of FastSizeDeque consistently yields superior performance.

%% \begin{figure}
%%     \centering
%%     \includegraphics[width=1.1\linewidth]{Ignite_YCSB.png}
%%     \caption{Performance Comparison of Different YCSB Workloads on Apache Ignite: FastSizeDeque vs. ConcurrentLinkedDeque}
%%     \label{fig:Ignite_YCSB}
%% \end{figure}

%% file: figures-mining-distribution.tex
\begin{tabular}{@{}l@{~~}l@{}}
  \begin{tikzpicture}[scale=0.8]
    \def\printonlylargeenough#1#2{\unless\ifdim#2pt<#1pt\relax
      #2\printnumbertrue
      \else
      \printnumberfalse
      \fi}
    \newif\ifprintnumber
    \pie[rotate=120, radius=1.3, before number=\printonlylargeenough{10}, after number=\ifprintnumber\%\fi]{
      26.6/\scriptsize get, 17.8/\scriptsize put, 13.1/\scriptsize remove, 42.5/\scriptsize \textit{\shortstack[c]{others\\(89)}}
    }
    \node () at (-.2,-2) {\bf\scriptsize\code{ConcurrentHashMap}};
  \end{tikzpicture}
  &
  \begin{tikzpicture}[scale=0.8]
    \def\printonlylargeenough#1#2{\unless\ifdim#2pt<#1pt\relax
      #2\printnumbertrue
      \else
      \printnumberfalse
      \fi}
    \newif\ifprintnumber
    \pie[rotate=120, radius=1.3, before number=\printonlylargeenough{10}, after number=\ifprintnumber\%\fi]{
      31.9/\scriptsize add, 20.8/\scriptsize remove, 19.6/\scriptsize contains, 27.7/\scriptsize \textit{\shortstack[c]{others\\(15)}}
    }
    \node () at (-.2,-2) {\bf\scriptsize\code{ConcurrentSkipListSet}};
  \end{tikzpicture}
  \\
  \begin{tikzpicture}[scale=0.8]
    \def\printonlylargeenough#1#2{\unless\ifdim#2pt<#1pt\relax
      #2\printnumbertrue
      \else
      \printnumberfalse
      \fi}
    \newif\ifprintnumber
    \pie[rotate=100, radius=1.3, before number=\printonlylargeenough{10}, after number=\ifprintnumber\%\fi]{
      28.8/\scriptsize add, 26.1/\scriptsize size, 11.4/\scriptsize poll, 33.7/\scriptsize \textit{\shortstack[c]{others\\(24)}}
    }
    \node () at (-.2,-2) {\bf\scriptsize\code{ConcurrentLinkedQueue}};
  \end{tikzpicture}
  &
  \begin{tikzpicture}[scale=0.8]
    \def\printonlylargeenough#1#2{\unless\ifdim#2pt<#1pt\relax
      #2\printnumbertrue
      \else
      \printnumberfalse
      \fi}
    \newif\ifprintnumber
    \pie[rotate=115, radius=1.3, before number=\printonlylargeenough{10}, after number=\ifprintnumber\%\fi]{
      36.9/\scriptsize get, 15.5/\scriptsize incAndGet, 14.1/\scriptsize set, 33.5/\scriptsize \textit{\shortstack[c]{others\\(131)}}
    }
    \node () at (-.2,-2) {\bf\scriptsize\code{AtomicLong}};
  \end{tikzpicture}
\end{tabular}

%% file: results-highcontention.tex
\begin{tabular}[t]{c c c c c}
    \begin{tikzpicture}[baseline={(current axis.north)}]
        \input{results-counter-counter.tex}
    \end{tikzpicture}
    &
    \begin{tikzpicture}[baseline={(current axis.north)}]
        \input{results-map-hashmap.tex}
    \end{tikzpicture}
    &
    \begin{tikzpicture}[baseline={(current axis.north)}]
        \input{results-map-skiplistmap.tex}
    \end{tikzpicture}
    &
    \begin{tikzpicture}[baseline={(current axis.north)}]
        \input{results-reference-reference.tex}
    \end{tikzpicture}
    &
    \begin{tikzpicture}[baseline={(current axis.north)}]
        \input{results-queue-queue.tex}
    \end{tikzpicture}
    \\
\end{tabular}

%% file: results-mixedworkload.tex
\begin{tabular}[t]{c c c c}
    \hspace{-0.05\textwidth}
    \begin{tikzpicture}
        \input{results-mixedworkload-25}
    \end{tikzpicture}
    &
    \begin{tikzpicture}
        \input{results-mixedworkload-50}
    \end{tikzpicture}
    &
    \begin{tikzpicture}
        \input{results-mixedworkload-75}
    \end{tikzpicture}
    &
    \begin{tikzpicture}
        \input{results-mixedworkload-100}
    \end{tikzpicture}
    \\
\end{tabular}

%% file: results-sizes.tex
\begin{axis}[
width=\linewidth,
bar width=0.03\linewidth,
height=0.25\textwidth,
ybar,
ylabel={Kops/s per thread},
ytick={0, 1.0E4, 2.0E4, 3.0E4, 4.0E4, 5.0E4, 6.0E4, 7.0E4, 8.0E4, 9.0E4, 1.0E5},
ymin=0,
ymax=100000,
% symbolic x coords={LeftMargin, 16k, Space, 32k, Space, 65k, RightMargin},
ymajorgrids,
axis x line*=bottom,
axis y line*=left,
xtick={0,1,2},
xticklabels={16K,32K,64K},
legend style={draw=none, at={(0.8,1.05)}, anchor=north, font=\footnotesize},
enlarge x limits=0.3,
% x=0.5\textwidth,
% scale only axis,
% xmin=LeftMargin,
% xmax=RightMargin,
cycle list={ {fill=orange}, {fill=orange} },
area legend,
ylabel style={font=\footnotesize},
xlabel style={font=\footnotesize},
title style={font=\footnotesize},
y tick label style={font=\footnotesize},
x tick label style={font=\footnotesize},]
\addplot+[nodes near coords, point meta=explicit symbolic, every node near coord/.append style={yshift=6pt, xshift=-0.1cm, rotate=45, anchor=west, font=\footnotesize}, error bars/.cd, y dir=both, y explicit] coordinates {
 (0,45903) +- (2423, 2422) [1]
 (1,60728) +-(638, 639) 
 (2,61738) +-(1317, 1318) 
};
\addplot+[nodes near coords, point meta=explicit symbolic, every node near coord/.append style={yshift=6pt, xshift=-0.1cm, rotate=45, anchor=west, font=\footnotesize}, error bars/.cd, y dir=both, y explicit] coordinates {
 (0,37935) +- (820, 819) [5]
 (1,25809) +-(6508, 6509) 
 (2,26759) +-(6754, 6753) 
};
\addplot+[nodes near coords, point meta=explicit symbolic, every node near coord/.append style={yshift=6pt, xshift=-0.1cm, rotate=45, anchor=west, font=\footnotesize}, error bars/.cd, y dir=both, y explicit] coordinates {
 (0,37240) +- (1422, 1423) [10]
 (1,25118) +-(6431, 6431) 
 (2,28874) +-(7549, 7549) 
};
\addplot+[nodes near coords, point meta=explicit symbolic, every node near coord/.append style={yshift=6pt, xshift=-0.1cm, rotate=45, anchor=west, font=\footnotesize}, error bars/.cd, y dir=both, y explicit] coordinates {
 (0,28060) +- (805, 804) [20]
 (1,28915) +-(3804, 3804) 
 (2,25416) +-(4863, 4864) 
};
\addplot+[nodes near coords, point meta=explicit symbolic, every node near coord/.append style={yshift=6pt, xshift=-0.1cm, rotate=45, anchor=west, font=\footnotesize}, error bars/.cd, y dir=both, y explicit] coordinates {
 (0,17765) +- (794, 794) [40]
 (1,20588) +-(918, 919) 
 (2,19646) +-(1304, 1305) 
};
\addplot+[nodes near coords, point meta=explicit symbolic, every node near coord/.append style={yshift=6pt, xshift=-0.1cm, rotate=45, anchor=west, font=\footnotesize}, error bars/.cd, y dir=both, y explicit] coordinates {
 (0,6870) +- (175, 176) [80]
 (1,10108) +-(363, 364) 
 (2,9778) +-(729, 729) 
};
\legend{\sys}
\end{axis}

\begin{axis}[
width=\linewidth,
bar width=0.03\linewidth,
height=0.25\textwidth,
ybar,
ytick={0, 1.0E4, 2.0E4, 3.0E4, 4.0E4, 5.0E4, 6.0E4, 7.0E4, 8.0E4, 9.0E4, 1.0E5},
ymin=0,
ymax=100000,
% symbolic x coords={LeftMargin, 16k, Space, 32k, Space, 65k, RightMargin},
ymajorgrids,
axis lines=none,
xtick={0,1,2},
xticklabels={16K,32K,65K},
legend style={draw=none, at={(0.785,0.88)}, anchor=north, font=\footnotesize},
enlarge x limits=0.3,
% x=0.5\textwidth,
% scale only axis,
% xmin=LeftMargin,
% xmax=RightMargin,
cycle list={ {fill=RoyalBlue!90!white}, {fill=RoyalBlue!90!white} },
area legend,
ylabel style={font=\Large},
xlabel style={font=\Large},
title style={font=\Large},
tick label style={font=\Large},
yticklabel={\empty},
xticklabel={\empty},
scaled y ticks=false,]
\addplot+[nodes near coords, point meta=explicit symbolic, every node near coord/.append style={yshift=3pt, xshift=-0.35cm, rotate=45, anchor=west, font=\Large}, error bars/.cd, y dir=both, y explicit]coordinates {
 (0,28858) +- (57, 57) 
 (1,41355) +-(350, 350) 
 (2,41189) +-(285, 286) 
};
\addplot+[nodes near coords, point meta=explicit symbolic, every node near coord/.append style={yshift=3pt, xshift=-0.32cm, rotate=45, anchor=west, font=\Large}, error bars/.cd, y dir=both, y explicit] coordinates {
 (0,13219) +- (139, 139) 
 (1,24491) +-(112, 112) 
 (2,28008) +-(181, 181) 
};
\addplot+[nodes near coords, point meta=explicit symbolic, every node near coord/.append style={yshift=3pt, xshift=-0.25cm, rotate=45, anchor=west, font=\Large}, error bars/.cd, y dir=both, y explicit] coordinates {
 (0,10053) +- (116, 115) 
 (1,18771) +-(136, 136) 
 (2,21928) +-(91, 91) 
};
\addplot+[nodes near coords, point meta=explicit symbolic, every node near coord/.append style={yshift=3pt, xshift=-0.14cm, rotate=45, anchor=west, font=\Large}, error bars/.cd, y dir=both, y explicit] coordinates {
 (0,7741) +- (43, 43) 
 (1,11830) +-(303, 304) 
 (2,14710) +-(257, 257) 
};
\addplot+[nodes near coords, point meta=explicit symbolic, every node near coord/.append style={yshift=3pt, xshift=-0.09cm, rotate=45, anchor=west, font=\Large}, error bars/.cd, y dir=both, y explicit] coordinates {
 (0,4950) +- (36, 36) 
 (1,5971) +-(66, 66) 
 (2,7634) +-(90, 89) 
};
\addplot+[nodes near coords, point meta=explicit symbolic, every node near coord/.append style={yshift=3pt, xshift=-0.1cm, rotate=45, anchor=west, font=\Large}, error bars/.cd, y dir=both, y explicit] coordinates {
 (0,1547) +- (18, 18) 
 (1,937) +-(48, 48) 
 (2,1246) +-(41, 40) 
};
\legend{JUC}
\end{axis}

%% file: results-retwis-results.tex
\begin{tikzpicture}
    \begin{axis}[
        ybar,
        bar width=0.02\linewidth,
        ylabel={Speedup},
        ytick={0,0.5,1,1.5,2},
        ymin=0,
        ymax=2,
        xlabel={\# threads},
        xticklabels={1,5,10,20,40,80,Avg,DAP},
        xtick={0,1,2,3,4,5,6,7},
        scale only axis,
        width=0.9\linewidth,
        height = 3cm,
        ymajorgrids,
        axis x line*=bottom,
        axis y line*=left,
        legend style={draw=none, at={(0.5,1.2)}, anchor=north},
        cycle list={
            {pattern color=orange!100!white, fill=orange!100!white},%, pattern=north east lines},
            {pattern color=orange!100!white, fill=orange!80!white, pattern=crosshatch},
            {pattern color=orange!100!white, fill=orange!60!white, pattern=north east lines},
            {pattern color=orange!100!white, fill=orange!40!white},%, pattern=grid}
        },
        area legend,
        ylabel style={font=\footnotesize},
        xlabel style={font=\footnotesize},
        title style={font=\footnotesize},
        tick label style={font=\footnotesize}
        ]

            \addplot+[nodes near coords, point meta=explicit symbolic, every node near coord/.append style={rotate=45, anchor=west, yshift=6pt, font=\footnotesize}, error bars/.cd, y dir=both, y explicit] coordinates {
                (0,0.9552) +- (0.0119,0.0119) [100K]
                (1,1.2071) +- (0.1358,0.1358)
                (2,1.4539) +- (0.1114,0.1114)
                (3,1.3379) +- (0.1383,0.1383)
                (4,1.4449) +- (0.0591,0.0591)
                (5,1.6442) +- (0.0585,0.0585)
                (6,1.3405) +- (0.0306,0.0306)
                (7,0)
            };

            \addplot+[nodes near coords, point meta=explicit symbolic, every node near coord/.append style={rotate=45, anchor=west, yshift=4pt, font=\footnotesize}, error bars/.cd, y dir=both, y explicit] coordinates {
                (0,1.0402) +- (0.0454,0.0454) [500K]
                (1,1.0158) +- (0.1198,0.1198)
                (2,1.1573) +- (0.0383,0.0383)
                (3,1.3029) +- (0.0637,0.0637)
                (4,1.3331) +- (0.0346,0.0346)
                (5,1.5725) +- (0.0591,0.0591)
                (6,1.237) +- (0.0201,0.0201)
                (7,0)
            };

            \addplot+[nodes near coords, point meta=explicit symbolic, every node near coord/.append style={rotate=45, anchor=west, yshift=6pt, font=\footnotesize}, error bars/.cd, y dir=both, y explicit] coordinates {
                (0,0.96) +- (0.0056,0.0056) [1000K]
                (1,1.0601) +- (0.1209,0.1209)
                (2,1.1258) +- (0.1176,0.1176)
                (3,1.2462) +- (0.0373,0.0373)
                (4,1.2787) +- (0.0258,0.0258)
                (5,1.5207) +- (0.0399,0.0399)
                (6,1.1986) +- (0.0172,0.0172)
                (7,0)
            };
            \draw[thick, black] (axis cs:-1,1) -- (axis cs:8,1);
            \draw[dashed, thick, black] (axis cs:5.5,0) -- (axis cs:5.5,3);
    \end{axis}

    \begin{axis}[
        ybar,
        bar width=0.02\linewidth,
        ytick={0,0.25,0.5,0.75,1,1.25,1.5,1.75,2},
        ymin=0,
        ymax=2,
        xticklabels={1,5,10,20,40,80,AvgAdjusted,AvgDAP},
        xtick={0,1,2,3,4,5,6,7},
        scale only axis,
        width=0.9\linewidth,
        height=0.25\textwidth,
        axis lines=none,
        legend style={draw=none, at={(0.5,1.2)}, anchor=north},
        cycle list={
            {pattern color=orange!100!white, fill=black!100!white},%, pattern=north east lines},
            {pattern color=black!100!white, fill=black!80!white, pattern=crosshatch},
            {pattern color=black!100!white, fill=black!60!white, pattern=north east lines},
            {pattern color=black!100!white, fill=black!40!white}%, pattern=grid}
        },
        area legend,
        yticklabel={\empty},
        xticklabel={\empty}
        ]
                
        \addplot+[error bars/.cd, y dir=both, y explicit,] coordinates {
            (0,0)
            (1,0)
            (2,0)
            (3,0)
            (4,0)
            (5,0)
            (6,0)
            (7,1.5722) +- (0.0931,0.0931) [black]
        };
                
        \addplot+[error bars/.cd, y dir=both, y explicit] coordinates {
            (0,0)
            (1,0)
            (2,0)
            (3,0)
            (4,0)
            (5,0)
            (6,0)
            (7,1.3747) +- (0.0576,0.0576) [black]
        };
                
        \addplot+[error bars/.cd, y dir=both, y explicit] coordinates {
            (0,0)
            (1,0)
            (2,0)
            (3,0)
            (4,0)
            (5,0)
            (6,0)
            (7,1.3572) +- (0.0544,0.0544) [black]
        };
    \end{axis}            
\end{tikzpicture}
\caption{
    \labfigure{retwis}
    Performance of the social network application relative to the baseline (\JUC).
}

%% file: results-retwis-alpha.tex
\begin{tikzpicture}
  \hspace{-2em}%
    \begin{axis}[
        xbar,
        xlabel={Throughput (Mops/s)},
        xtick={2000, 4000, 6000, 8000, 10000},
        xticklabels={2, 4, 6, 8, 10},
        xmin=0,
        xmax=10000,
        height=0.9\textwidth,
        bar width=0.06\textwidth,
        width=\textwidth,
        ylabel={$\alpha$},
        yticklabels={0.01,0.1,1},
        ytick={0,1,2},
        scaled ticks=false,
        enlarge y limits=0.2,
        scale only axis,
        axis x line*=bottom,
        axis y line*=left,
        legend style={draw=none, at={(1.1,0.45)}, anchor=north, font=\footnotesize},
        cycle list={
            {fill=RoyalBlue!90!white},
            {fill=black},
            {fill=orange,}
        },
        area legend,
        ylabel style={font=\footnotesize},
        xlabel style={font=\footnotesize},
        tick label style={font=\footnotesize}]
            \addplot+[error bars/.cd, x dir=both, x explicit] coordinates {
                (6428.3667,0) +- (69.89366300798899,69.89366300798899)
                (6527.4667,1) +- (57.73432605873163,57.73432605873163)
                (7437.4,2) +- (83.62435743938141,83.62435743938141)
            };
                    
            \addplot+[error bars/.cd, x dir=both, x explicit] coordinates {
                (7400.8333,0) +- (45.160843219613525,45.160843219613525)
                (7500.8667,1) +- (50.92175033725727,50.92175033725727)
                (9371.1333,2) +- (68.46899798110381,68.46899798110381)
            };
                    
            \addplot+[error bars/.cd, x dir=both, x explicit] coordinates {
                (7771.7,0) +- (64.23452352517101,64.23452352517101)
                (7882.4333,1) +- (70.08112143340733,70.08112143340733)
                (9713.6333,2) +- (102.19860654484881,102.19860654484881)
            };
                    
            \legend{JUC,DAP,\sys}
    \end{axis}
\end{tikzpicture}
\caption{
    \labfigure{alpha}
    Varying the user access distribution.
}

%% file: related.tex
\section{Related Work}
\labsection{related}

% commutativity
Programming multicore hardware is notoriously difficult \cite{broken}.
To simplify its use, some recent works propose to leverage commutativity.
With Veracity~\cite{veracity}, programmers write sequential code then annote the program to express conditions under which code fragments commute.
% These conditions are automatically verified and even inferred by a compiler.
%
In \cite{rule}, the authors introduce the SIM-commutativity rule.
% FIXME say a word about the sufficient part of the rule, which require controlling the scheduler
% The rule is based on the notion of SIM-commutativity.
A sequence $y$ SIM-commutes in a sequential history $h$ if for any sub-sequence $x$ of $y$, the permutations of $x$ are all equivalent.
\citet{rule} show that if $y$ SIM-commutes in $h$ then there exists an implementation of $y$ such that the steps of the threads when executing $y$ are conflict free.
Translated in our framework, the rule defines a sufficient condition for $y$ to be strongly labeling in $G(y,s)$, where $s$ is the state before entering $y$.
The commuter tool in \cite{rule} checks a specification for non-commuting operations.
This corresponds to verifying the sufficiency part in \refprop{longlived}.

% jdk
Some adjusted objects in \sys rely internally on a segmentation.
In the JDK, constructs such as the ones used in \code{LongAdder} and \code{ConcurrentHashMap}\footnote{prior to Java 8} work similarly.
However, they differ from ours in the sense that each segment is writable by multiple threads.
(In a segmentation, a segment is \SWMR.)
Permitting this requires complex concurrency-control mechanisms with a non-negligible performance cost.

% consensus number
Jayanti~\cite{jayanti96} proposes a hierarchy ($h_m^r$) to classify objects based on their consensus numbers.
The hierarchy is robust when no collection of objects with consensus number $n$ can be used to solve consensus in an $(n+1)$-process distributed system.
Several works explore the consensus hierarchy, for instance to find objects at each level \cite{daianlag18}, or to test its robustness \cite{kleinbergm92}.
In particular, Ruppert \cite{ruppert00} studies two classes of objects: read-modify-write (RMW) and readable.
The RMW class contains objects whose operations change the state of the object before returning the previous state.
Readable types are objects for which there exists an operations that can read the state without changing it.
He provides a characterization of the consensus number for objects in these classes.
\reftheo{distcn} (in \refsection{dist}) bears similarities with these results.
% We believe that our characterization is simpler. % FIXME

% data struct.
Many works explore the design of efficient shared objects and data structures
We refer the reader to Herlihy and Shavit’s book \cite{art} for detailed examples.
\citet{sela22} propose an efficient construction to maintain the size of a data structure.
% All these works are complementary to our.
The work in \cite{laws} shows that objects which provide strongly non-commuting operations at their interfaces are inherently non-scalable.
% This characterization does not take into account that concurrent operations may commute.

% weak consistency
Weak consistency is a possible approach to improve parallel programs.
\citet{QuasiLinearizability} permit concurrent histories that are not legal (linearizable) but that are at a bounded distance from a legal history.
Another relaxed model is IVL which allows objects to return values that are close to correct ones \cite{IVL}.
Some works change the level of consistency for operations upon request.
For instance, this is the case with lazy replication \cite{lazy} and RedBlue \cite{RedBlue}.
Anna~\cite{anna} is a key-value store built atop CRDTs.
Weak consistency has the potential to scale a program.
On the other hand, it complexifies coding and does not permit to maintain certain application invariants.

%% file: conclusion.tex
\section{Closing remarks}
\labsection{conclusion}

\paragraph{Limitations}
In \refsection{dist}, we provide a characterization of scalability based on the indistinguishability graph.
There are currently two issues with these results that we plan to address in a near future.
First, as in \cite{rule}, the existence of a universal construction with limited (or no) conflicts does not ensure that it is practical with current hardware.
Second, reducing conflict may translate into write amplification.
As noted in \refsection{evaluation}, this can deteriorate the CPU cache which in turn may void the benefits of lower contention.

\paragraph{Future work}
A next logical step is to generate indistinguishability graphs from the specification.
This would guide programmers to choose and design objects for their applications.
The results reported in \reffiguretwo{mining}{pie} were found with the help of scripts.
They provide information about the methods called in a program and how their return values are used.
Going further requires to find objects that bottleneck a program.
In the present paper, this question is addressed using expertise.
Automating this step demands to analyze the code at runtime with a profiler \cite{BouksiaaTLVDGBT19}.

\paragraph{Conclusion}
Adjusting an object combines subtyping and access restriction to its interface.
As captured by the indistinguishability graph, this reduces conflicts among threads and makes the object potentially more scalable.
Adjusted objects are already in the wild.
This paper introduces their formal foundations.
It also presents a library of adjusted objects for Java called \sys.
\sys provides drop-in replacements for the standard concurrent objects from the JDK.
In our evaluation, these adjusted objects can be up to two order of magnitude faster that the common ones.

%% file: appendix-model.tex
\section{System Model}
\labappendix{model}

Our system model follows the standard definitions in \cite{Herlihy}.
The entities that perform operations are called \textit{threads}.
There are $n \geq 2$ of them.
They communicate through shared data, or \textit{objects}.
Threads and objects are simplified deterministic I/O automata \cite{Automata}.
Below, we recall the elements of this system model and introduce some key notions to state our results.

\paragraph{Object}
An object is defined by a \textit{data type}.
The data type models the possible states of the object, the operations to access it, as well as the response values from these operations.
Formally, a (sequential) data type is an automaton $A=(\stateSet, \stateInit, \cmdSet, \valSet, \tau)$ where $\stateSet$ is the set of states of $A$, $\stateInit \in \stateSet$ its initial state, $\cmdSet$ the operations of $A$, $\valSet$ the response values, and $\tau : \stateSet \times \cmdSet \rightarrow \stateSet \times \valSet$ defines the transition relation.
An operation $c$ is \emph{total} if $\stateSet \times \{c\}$ is in the domain of $\tau$.
Operation $c$ is \emph{deterministic} if the restriction of $\tau$ to $\stateSet \times \{c\}$ is a function.
Hereafter, we assume that all operations are total and deterministic.
The selections $.\st$ and $.\val$ extract respectively the state and the response value components of an operation, that is given a state $s$ and an operation $c$, $\tau(s,c)=(\tau(s,c).\st,\tau(s,c).\val)$.
Function $\tau^{+}$ is defined by repeating the application of $\tau$, that is for some sequence of operations $\sigma=c_1{\ldots}c_{k \geq 1}$ and a state $s$ we have:
\begin{displaymath}
  \tau^{+}(s,\sigma)
  \equaldef
  \left\lbrace
  \begin{array}{ll}
    \tau(s,c_1) & \text{if $k=1$,} \\
    \tau^{+}(\tau(s,c_1).\st,c_2{\ldots}c_k) & \text{otherwise.}
  \end{array}
  \right.
\end{displaymath}
For some sequence of operations $\sigma = c_1 \ldots c_k$, we write $s.\sigma$ the state reached after applying the operations of $\sigma$ in order.
The response of $c_i \in \sigma$ from $s$ is defined as $\tau_v(s.\sigma', c_i)$, where $\sigma'$ is the prefix of $\sigma$ before $c_i$.

\paragraph{Hoare logic}
We use Hoare's notations to specify objects.
An operation $c$ is defined as a tuple $\hoare{P}{c}{Q}$, where $P$ and $Q$ are respectively the pre- and postconditions of $c$, that is predicates over the state of the object.
When $c$ is executed in some state that satisfies $P$ and at the end of the execution of $c$, $Q$ is satisfied, the triplet $\hoare{P}{c}{Q}$ is true.
Otherwise, when $c$ is applied to a state which does not satisfy $P$ (or $Q$), $c$ does not alter the object and $c$ returns the special value $\bot$.
In such a case, the triplet $\hoare{P}{c}{Q}$ is false.
Below, we use such notations to specify a counter with operations increment, get, and reset.
As standard, $s$ is the state of the object and $s'$ is the new state after the operation is applied.
We note $\ret$ the operation's response value.
\begin{itemize}
\item[-] \hoare{\true}{\inc()}{s' = s+1 \land \ret = s' }
\item[-] \hoare{\true}{\get()}{\ret = s}
\item[-] \hoare{\true}{\reset()}{s' = 0} 
\end{itemize}

\paragraph{Liskov's substitution principle}
In \cite{subtyping}, Liskov and Wing define a substitution principle that permits to replace a type with another.
A data type $S$ is a \textit{subtype} of $T$ when the following conditions hold.
(1) There is an abstraction function that maps a state of the subtype to a state of the supertype which maintains any state invariant.
(2) $S$ preserves the operations of $T$.
This means that if $\hoare{P}{c_S}{Q}$ of $S$ corresponds to the operation $\hoare{P'}{c_T}{Q'}$ of $T$, then it is true that:
(2a) \emph{Contravariance of arguments}: $c_s$ and $c_T$ have the same number of arguments and denoting respectively $\alpha$ and $\beta$ these lists of arguments, for any $i$, $\alpha[i]$ subtypes $\beta[i]$;
(2b) \emph{Covariance of result}: Either both $c_S$ and $c_T$ have a result or neither has, and the type of the response value of $c_T$ subtypes the one of $c_S$;
(2c) \emph{Exception rule}: The exceptions signaled by $c_S$ are contained in the set of exceptions signaled by $c_T$;%
\footnote{Exception are simply response values in our model.}
(2d) \emph{Pre-condition rule}: $c_S$ can be called at least in any state required by $c_T$ ($P' \implies P$);
and
(2e) \emph{Post-condition rule}: $c_S$'s post-condition is stronger than $c_T$'s ($Q \implies Q'$).
Finally, (3) \emph{Constraint rule:} Any sequence of state changes observed in the subtype should be valid for the supertype.
This substitution principle ensures that if a property $\phi(O)$ is true for any object $O$ of type $T$, then $\phi(O')$ is true for any object $O'$ of type $S$.

\paragraph{Histories}
A history is a sequence of invocations and responses of operations by the threads on one or more objects.
We write $c \hb_h d$ when operation $c$ precedes $d$ in history $h$.
This order is called \emph{happens-before}.
Operations unrelated by $\hb_h$ are said to be \emph{concurrent}.
Following \citet{Herlihy}, histories have various properties according to the way invocations and responses interleave.
For the sake of completeness, we recall these properties in what follows.
A history $h$ is \emph{complete} if every invocation has a matching response.
A \emph{sequential} history $h$ is a non-interleaved sequence of invocations and matching responses, possibly terminated by a non-returning invocation.
When a history $h$ is not sequential, we say that it is \emph{concurrent}.
A history $h$ is \emph{well-formed}
if
(i) $h|p$ is sequential for every thread $p$,
(ii) for every operation $c$, $c$ is invoked at most once in $h$,
and (iii) for every response to some operation $c$, there exists an invocation of $c$ by the same thread that precedes it in $h$.
A well-formed history $h$ is \emph{legal} if for every object $O$, $h|O$ is both complete and sequential, and denoting $c_1{\ldots}c_{n \geq 1}$ the sequence of operations appearing in $h|O$, if for some operation $c_k$ a response value appears in $h|0$, it equals $\tau^{+}(\stateInit, c_1{\ldots}c_k).\val$.

\paragraph{Linearizability}
Two histories $h$ and $h'$ are said \emph{equivalent} if they contain the same set of events.
Given a history $h$, $h$ is \emph{linearizable} \cite{Herlihy} if it can be extended (by appending zero or more responses) 
to some complete history $h'$ equivalent to a legal and sequential history $l$ with $\hb_{h'} \subseteq \hb_{l}$.
In such a case, history $l$ is named a \emph{linearization} of $h$.

\paragraph{Wait-freedom}
Since we work in a shared memory system, we want one thread to be able to finish its tasks regardless of the activity of the other threads.
This property is called \textit{wait-freedom} \cite{Herlihy}.

\paragraph{Distributed task}
Recall that there are $n$ threads in the system.
A distributed task $\Delta$ is defined by a set $\mathcal{I}$ of input $n$-dimensional vectors, a set $\mathcal{O}$ of output $n$-dimensional vectors and a map $\Delta$ from $\mathcal{I}$ to $2^{\mathcal{O}}$.
A distributed algorithm $\mathcal{A}$ (i.e., one automaton per thread using shared objects) solves a task $\Delta$ \cite{Raynal18} when for every run of $\mathcal{A}$ there exists $(I,O) \in \Delta$ such that
\begin{inparaenum}
\item each thread $p$ starts with $in_p = I[p]$, and
\item each thread $p$ that computes (if any) some output $out_p$ satisfies $O[p] = out_p$.
\end{inparaenum}

%% file: appendix-proofs.tex
\section{Proofs}
\labappendix{proofs}

\subsection{Relation with consensus}
\labappendix{proofs:consensus}

\distcn*

\begin{proof}
  The theorem is established using two successive inequations.
  First, we prove that if $T$ can solve consensus among $k \geq 2$ threads, then there must exist a state $s$ and a bag of operations $B$ such that $\cardinalOf{B}=k$ and $\Gr(B,s)$ has at least two connected components.
  This shows that the consensus number of $T$ is upper bounded by $\max \{k : \exists l \geq 2 \sep T \in \dist{k,l}\}$.
  Then, we establish that binary consensus is solvable among $n \geq 2$ threads under the assumption that for some bag of operations $B$ with $\cardinalOf{B}=n$ and some state $s$, $\Gr(B,s)$ contains at least two components.
  As a consequence, $\max \{k : \exists k \geq 2 \sep T \in \dist{k,l}\}$ is capped by the consensus number of $T$.
  
  \paragraph{($\leq$)}
  Assume an algorithm $\mathcal{A}$ solving binary consensus among $n \geq 2$ threads, using any number of objects of type $T$ and registers.
  Our reasoning relies on the standard valency argument introduced in FLP \cite{FLP}.
  In a nutshell, this reasoning goes as follows:
  We consider all the runs of $\mathcal{A}$.
  In a run, the collective states of the threads and the shared memory at some point in time is called a \emph{configuration}.
  Configurations of all the runs of $\mathcal{A}$ are connected to form a directed graph.
  There exists an edge $e$ labeled with $p$ between configurations $\mathcal{C}$ and $\mathcal{D}$ when the next operation of thread $p$ in $\mathcal{C}$ leads to configuration $\mathcal{D}$.
  If $c$ is the next operation called by a thread in configuration $\mathcal{C}$, then $\mathcal{C}.c$ is the configuration reached after applying it.
  Consider a configuration $C$ and let $V$ be the set of decision values of configurations reachable from $C$.
  If $\cardinalOf{V}=1$, $C$ is said \emph{monovalent}.
  When this happens, $C$ is $0$-valent or $1$-valent according to the corresponding decision value.
  The configuration is \emph{bivalent} when $\cardinalOf{V}=2$.
  When $C$ is bivalent and all the configurations immediately after it are monovalent, $C$ is \emph{critical}.

  By applying a classical reasoning~\cite{Herlihy16}, one can show that there must exist a critical configuration $\mathcal{C}$.
  Moreover, in configuration $\mathcal{C}$, all the threads are about to access the same object, say $O$, of type $T$.
  Let $s$ be the state of object $O$ in $\mathcal{C}$.
  Note $B$ the bag of the $n$ operations threads are about to execute.
  
  For $u \in \{0,1\}$, there exists $x_u \in \perm(B)$ such that for any non-empty $z \pref x_u$, $z$ leads to a $u$-valent configuration from $\mathcal{C}$.
  Indeed, from $S$ being bivalent, there exists two operations $c_0,c_1 \in B$ bringing $s$ to respectively a $0$-valent and a $1$-valent configuration.
  Hence, any permutation of $B$ starting with $c_0$ (resp., $c_1$) leads to a $0$-valent (resp., $1$-valent) configuration.
  We pick $x_0$ and $x_1$ among these permutations.

  Next, we observe that if $\Gr(B,s)$ has a single component then $x_0$ and $x_1$ have the same valency.
  If $\clazz{x_0} = \clazz{x_1}$ then there exists successive permutations $y_1, \ldots, y_{k\geq{2}}$ such that
  \begin{inparaenum}
  \item $y_1=x_0$ and $y_{k}=x_1$, and
  \item for each $i \in [1,k]$, $y_i \indistinguishable{c_i,s} y_{i+1}$ for some operation $c_i$.
  \end{inparaenum}
  From $y_i \indistinguishable{c_i,s} y_{i+1}$, operation $c_i$ returns the same response in $y_i$ and $y_{i+1}$.
  Moreover, there exists a common state posterior to $c_i$ in both $y_i$ and $y_{i+1}$.
  Let $z_i$ and $z_{i+1}$ be the corresponding prefixes, and name $p$ the thread about to execute $c_i$ from configuration $\mathcal{C}$.
  Both $\mathcal{D}=\mathcal{C}.z_i$ and $\mathcal{D}'=\mathcal{C}.z_{i+1}$ are indistinguishable for $p$.
  As a consequence, any solo run of $p$ starting from $\mathcal{D}$ is also applicable to $\mathcal{D}'$
  Hence, they have the same valency.
  By induction, we deduce that $x_0$ and $x_1$ have the same valency;
  contradiction.
  
  \paragraph{($\geq$)}
  The result is obtained by successive reductions.
  We rely on two variations of the binary consensus problem, namely weak consensus and team consensus.
  In weak consensus, validity is replaced with a weaker property:
  there exists a run in which $0$ is decided and a run in which $1$ is decided.
  Team consensus splits the threads in two (known a priori) teams.
  The agreement property of consensus is replaced with:
  if all threads in a team have the same input value, then no two threads decide different values.
  From weak consensus, one may obtain team consensus and then from team consensus, (binary) consensus itself.
  These two reductions are known result from literature \cite{petr-phd}.
  We show that under the hypothesis that two components exist, weak consensus is solvable.

  Assume that for some bag of operations $B$ with $\cardinalOf{B}=n$ and some state $s$, $\Gr(B,s)$ contains at least $2$ components.
  To solve weak consensus, we use a single shared object $O$ of type $T$.
  The algorithm works as follows.
  
  \begin{construction}
    Initially, object $O$ is in state $s$.
    From \cite{sense}, it is possible to make such an assumption.
    We map each thread $p$ to a unique operation $c_p \in B$, while respecting the permission map $O.m$.
    Similarly, each indistinguishability class $\clazz{z}$ in $\Gr(B,s)$ is mapped to some value $d(\clazz{z}) \in \{0,1\}$.
    The map $d$ is surjective, which is possible because by assumption there are at least two classes.
    Upon proposing to consensus, thread $p$ calls $c_p$, returning some result $r$ from that call.
    Then it reads object $O$, retrieving a state $s'$.
    There must exist $x \in \perm(B)$ such that $c_p$ returns $r$ in $\tau(s,x)$ and state $s'$ follows $c_p$ in $\tau(s,x)$.
    Thread $p$ decides $d(\clazz{x})$.
  \end{construction}

  We now establish that the above construction is correct.
  Consider some run and let $l$ be a linearization of the operations over $O$ in this run.
  Pick some permutation $z$ of the operations in $B$ such that $l|B$ prefixes $z$.
  We prove that any thread $p$ that decides must decide $d([z])$.
  Let $r$ be the result of operation $c_p$, $s'$ the state read by $p$, and $x$ be the permutation computed by $p$.
  It is the case that $c_p$ returns $r$ in $\tau(s,x)$ and state $s'$ follows $c_p$ in $\tau(s,x)$.
  By definition of $l$, this is also true in $l|B$.
  Hence, $x$ and $z$ are indistinguishable from $s$ for $c_p$.
  This leads to $\clazz{x} = \clazz{z}$, as required.
  For $u \in \{0,1\}$, there exists one indistinguishability class $\clazz{z_u}$ with $d(\clazz{z_u})=u$.
  From what precedes, for some permutation $x \in \clazz{z_u}$, if $x[0]$ is invoked solo in a run, $u$ is decided.
  This show that weak validity holds.
\end{proof}

In the proof above, we rely on the fact that a thread may read the object after applying a write operation.
This requires the object to be long-lived.
If we restrict our attention to one-shot objects, as the state at the end of a permutation does not matter, showing the result is simpler.
In detail, the indistinguishability relation between elements in $\perm(B)$ is defined only using return values.
The ($\geq$) part of \reftheo{distcn} is unchanged.
For the ($\leq$) part, the thread does not read the state of the object but just uses its return value to locate the component.

To complete \reftheo{distcn}, we also provide a full picture of the readable objects in $\cn_1$.
Recall that  $T$ is historyless when all its write operations are overwriting each other \cite{jayanti96}.
It is well-known that these objects are in $\cn_1$.
This is also the case if the write operations are all commuting.
Two operations are \emph{weakly-commuting} when in every state, applying them in any order leads to the same new state, and one operation does not notice the other.
For instance, a read and a write are weakly-commuting.
This is also the case of an increment and a fetch-and-add.
$T$ is \emph{permissive} when all pairs of write operations are either overwriting or weakly-commuting.
The result below establishes that the permissive objects are exactly the objects in $\cn_1$ which are jointly linearizable, deterministic, and readable.
This result was found independently by \citet[Theorem 1]{cn1}.

\begin{corollary}
  \lablem{proofs:permissive}
  Consider $T$ a readable type.
  $T$ is in $\cn_1$ if and only if $T$ is permissive.
 \end{corollary}

\begin{proof}
  (\implies)
  For the sake of contradiction, assume that $T$ is not permissive.
  This means that some pair of operations is neither overwriting, nor weakly-commuting.
  In other words, there exists $(s,c,d)$ such that:
  
  $$
  \begin{array}{ll}
    & \neg ~ \lor ~ (\tau(s,c) = \tau(s.d,c)) \\
    & \hspace{1.2em} \lor~ (\tau(s,d) = \tau(s.c,d)) \\
    & \hspace{1.2em} \lor~\land (\tau(s.c,d).\st = \tau(s.d,c).\st) \\
    & \hspace{2.4em} \land~\lor~(\tau(s,c).\val=\tau(s.d,c).\val) \\
    & \hspace{3.5em} \lor \hspace{0.5em} (\tau(s.d).\val=\tau(s.c,d).\val) \\
    = & \land ~ \lor ~(\tau(s,c).\st \neq \tau(s.d,c).\st) \\
    & \hspace{1.2em} \lor \hspace{0.5em} (\tau(s,c).\val \neq \tau(s.d,c).\val) \\
    & \land ~ \lor ~ (\tau(s,d).\st \neq \tau(s.c,d).\st) \\
    & \hspace{1.2em} \lor \hspace{0.5em} (\tau(s,d).\val \neq \tau(s.c,d).\val) \\
    & \land ~ \lor ~ (\tau(s.c,d).\st \neq \tau(s.d,c).\st) \\
    & \hspace{1.2em} \lor \hspace{0.5em} \land (\tau(s,c).\val \neq \tau(s.d,c).\val) \\
    & \hspace{2.6em} \land \hspace{0.3em} (\tau(s,d).\val \neq \tau(s.c,d).\val) 
  \end{array}
  $$
  
  We show that this permits to solve consensus for two threads $p$ and $q$.
  For starters, consider that ($\tau(s,c).\val \neq \tau(s.d,c).\val$) and ($\tau(s.d).\val \neq \tau(s.c,d).\val$) hold.
  In order to solve consensus, we use one register per thread and some object $O$ of type $T$, initially in state $s$.
  Thread $p$ (respectively $q$) writes its proposal in its register then applies operation $c$ (respectively $d$) to $O$.
  When returning from this call, $p$ may identify whether $c$ is linearized before $d$ or not.
  If this holds, it decides its own proposal.
  In the converse case, the thread picks the other proposal.
  Otherwise, $\tau(s.c,d).\st \neq \tau(s.d,c).\st$.
  In this case, the construction is similar to the previous one, but requires in addition to read the state of the object.
  This read serves to determine the linearization order.

  (\impliedby)
  Consider a permissive object $O$.
  Any indistinguishability graph in $\Gr_O$ has a single component.
  Applying~\reftheo{distcn}, $O$ is in $\cn_1$.
  %% *long version*
  %% Suppose that there exists a two-thread consensus protocol implemented from permissive objects and atomic read/write registers.
  %% As in the proof of \reftheo{distcn}, let $\mathcal{C}$ be a critical configuration such that the two threads are poised to access the same permissive object.
  %% Thread $p$ is about to execute an operation $c$ that leads to a (say) $0$-valent state, while the operation $d$ of the other thread $q$ leads to a $1$-valent state.
  %% %
  %% First, consider that an operation overwrites the other:
  %% Let $\mathcal{D}$ be the $0$-valent configuration if $c$ is applied then $d$ and $\mathcal{D'}$ be the $1$-valent configuration when $d$ is executed.
  %% Since $d$ overwrites the value written by $c$, $\mathcal{D}$ and $\mathcal{D'}$ only differ in the internal state of thread $p$.
  %% Therefore $q$ decides the same value starting from $\mathcal{D}$ and $\mathcal{D'}$, which is impossible because the two configurations have different valency.
  %% %
  %% Next, suppose that the two operations weakly-commute:
  %% Let $s$ be the state of the considered object in configuration $\mathcal{C}$.
  %% Without lack of generality, assume that $\tau(s.d,c).\val=\tau(s,c).\val$ and $\tau(s.c,d).\st=\tau(s.d,c).\st$.
  %% Let $\mathcal{D}$ be the $0$-valent configuration if $c$ is executed then $d$ and $\mathcal{D'}$ be the $1$-valent configuration if the two operations are executed in converse order.
  %% From what precedes, $\mathcal{D}$ and $\mathcal{D'}$ are not distinguishable for thread $p$;
  %% a contradiction.
\end{proof}

\subsection{Predicting scalability}
\labappendix{proofs:scalability}

Below, we state a technical lemma then prove our results linking the scalability of an object to its indistinguishability graphs.

\begin{lemma}
  \lablem{distlab}
  Consider a shared object $O$ whose data type is $T$, some state $s$ of $T$, and a bag $B$ of operations.
  If $c$ is labeling in $\Gr_T(B,s)$
  then
  for every permutation $x \in B$,
  the response value of $c$ in $x$ is the response obtained when applying $c$ to state $s$.
\end{lemma}

\begin{proof}
  Consider some permutation $x$ of $B$ and let $v$ be the response value of $c$ in $x$.
  Let $y$ be a permutation of $B$ starting with $c$.
  As $c$ is labeling in $\Gr_T(B,s)$, there is an edge $(x,y)$ in $\Gr_T(B,s)$ with a label including $c$.
  Thus, the response of $c$ in $x$ is the same as in $y$.
  Since $c$ is the first operation in $y$, its response is $\tau(s,c).\val$.
\end{proof}

\oneshot*

\begin{proof}
  \begin{inparaenumorig}
  \item[($\impliedby$)]
    To implement $O$ in a conflict-free manner, we proceed as follows:
    Each thread holds a local copy of the object, initially in state $s$.
    Upon calling an operation (at most once because the object is one-shot), the thread applies the operation to its local copy then it returns the corresponding response value.
    Clearly, this implementation is both wait-free and conflict-free.
    Linearizability follows from \reflem{distlab} and the fact that $B$ is labeling in $\Gr_T(B,s)$ for any bag $B$.
    In detail, consider a run $\run$ in which threads execute the operations in $B$.
    Let $x$ be a permutation of $B$ that respects the real-time order in $\run$.
    By assumption, $B$ is labeling in $\Gr_T(B,s)$.
    Hence, applying \reflem{distlab}, the response value of $c \in B$ is $\tau(s,c).\val$.
    From which it follows that $x$ is a linearization of $\run$.    
  \item[($\implies$)]
    Let $\mathcal{I}$ be a conflict-free and wait-free implementation for $O$.
    Consider some bag $B$ of $\cardinalOf{\procSet}$ operations, and two permutations $x,y \in \perm(B)$.
    Let $\run_x$ (respectively, $\run_y$) be a run of $\mathcal{I}$ in which all the threads execute the operations in $B$ following the order in $x$ (resp. $y$).
    This is possible because $B$ must comply with $O.m$ and $\mathcal{I}$ is wait-free.
    Consider an operation $c \in B$, and let $d$ be an operation before $c$ in $x$.
    Because $\mathcal{I}$ is conflict-free, $c$ never reads a register written by $d$.
    Hence the response value of $c$ in $x$ is $\tau(s,c).\val$.
    The same observation holds for permutation $y$.
    Moreover, because operations do not write to a common register, the object is in the same state after $\run_x$ and $\run_y$. % FIXME ?
    From which it follows that $c$ labels $(x,y)$ in $\Gr_T(B,s)$.
  \end{inparaenumorig}
\end{proof}

\longlived*

\begin{proof}
  We examine each side of the equivalence.
  \begin{inparaenumorig}
  \item[($\impliedby$)]
    The implementation is exactly the same as in the one-shot case.
    That is, each thread maintains a local copy of the object.
    Upon executing an operation $c$, the operation is applied to the local copy and the corresponding response value returned.
    Let $\run$ be a run of this implementation and $h$ be the history induced by $\run$.
    As the above algorithm is wait-free, we assume without lack of generality that $h$ is complete.
    Consider some linearization of $(\mathit{ops}(h),\ll_h)$, where $\mathit{ops}(h)$ are the operations in $h$ and $\ll_h$ is the real-time order in $h$.
    Let $l$ be the sequential history in which operations are executed according to this linearization.
    We establish that all the operations in $l$ have the exact same responses as in $h$.
    This shows that history $h$ is linearizable.
    For this, we iterate over all the operations, one thread at a time, considering the order in which they appear in $l$.
    Each operations is ``validated'' by moving it at the start of the sequence, as in run $\run$.
    In detail, consider some thread $p$.
    Let $c$ be the last non-validated operation by $p$, and $d$ be the operation right before it in $l$.
    If $d$ is also executed by $p$ we are done:
    the start of the sequence consists in all the operations of $p$ which were executed before $c$.
    This corresponds exactly to what happens in $\run$.
    Otherwise $c$ and $d$, can be swapped while preserving the same response values everywhere:
    by assumption, $c$ and $d$ have the exact same response before the swap (they are both labeling), and the swap does not change the state of the object (labels are strong).
  \item[($\implies$)]
    Consider an indistinguishability graph $\Gr_T(B,s)$ in $\Grs_O$, with $B=\{c,d\}$ a bag of operations.
    We prove that if a conflict-free implementation $\mathcal{I}$ of $O$ exists, $B$ is strongly labeling in $\Gr_T(B,s)$.
    We reach state $s$ by first executing an appropriate sequential run $\run_0$.
    This is possible because every operation is executable by at least one thread.
    Then, consider that two thread $p$ and $q$ are about to execute respectively operation $c$ and $d$.
    (With $c \in O.m[p]$ and $d \in O.m[q]$.)
    Choose a continuation in which $p$ executes first $c$ followed by $d$ by $q$.
    Because $\mathcal{I}$ is conflict-free, the return value of $d$ is the same as when applying $d$ first.
    Hence $d$ is labeling in $\Gr_T(B,s)$.
    Using a symmetrical argument, $c$ is also labeling.
    Moreover, the state of the object is exactly the same after applying $cd$ or $dc$.
  \end{inparaenumorig}
\end{proof}

\input{figures-power}

\begin{lemma}
  \lablem{proofs:power}
  If $c$ has consensus power
  then
  there is no update-conflict free implementation of $c$.
\end{lemma}

\begin{proof}
  Let $T$ be the data type of operation $c$.
  Assume that $c$ has an update-conflict free implementation $\mathcal{I}$.
  Let $\mathcal{A}$ be an algorithm solving binary consensus with objects of type $T$.
  We consider that all these objects are implemented with $\mathcal{I}$.
  In what follows, we construct a run $\rho$ of $\mathcal{A}$ in which two correct threads $p$ and $q$ never terminate.

  The construction of $\rho$ is illustrated in \reffigure{power}.
  Starting from some initial bivalent configuration $\mathcal{C}_0$ of $\mathcal{A}$, $\rho$ must reach a critical bivalent configuration $\mathcal{C}$ when alternating the steps of $p$ and $q$.
  Otherwise, the run never terminates and we are done.
  In $\mathcal{C}$, $p$ and $q$ are about to do steps $\step_p$ and $\step_q$ on registers $X_p$ and $X_q$, respectively.
  First of all, observe that these steps are within a call to $c$ on the same object of type $T$ (left side of \reffigure{power}).
  Otherwise, we can apply the usual reasoning about the impossibility of consensus in read/write shared memory.
  Then, depending on the nature of the steps, there are several cases to consider (right side of \reffigure{power}).
  \begin{inparaenum}
  \item
    Assume that $\step_p$ is a read over $X_p$.
    Consider a solo continuation $\lambda$ of $q$ from $\mathcal{C}$ during which $q$ decides.
    Then, $\lambda$ is also applicable from $\mathcal{C}'=\mathcal{C}.\step_p$, leading to the same decision;
    contradiction.
    The very same reasoning holds in case $\step_q$ is a read.
  \item
    Assume that both steps update the registers.
    Then, because $\mathcal{I}$ is update-conflict free, these two steps cannot be on the same register.
    As a consequence, these two steps are commuting.
    Hence, the valency of $\mathcal{C}.\step_p.\step_q$ is the same as $\mathcal{C}.\step_q.\step_p$, which contradicts that $\mathcal{C}$ is critical.
  \end{inparaenum}
\end{proof}

\updateconflict*

\begin{proof}
  The ``only if'' part follows from \reflem{proofs:power}.
  Regarding the ``if'' part, below we detail and prove correct an update-conflict free implementation.
  \begin{construction}
    Threads share a global clock.
    Each thread also holds a copy of the object, a log and a local clock.
    Entries in the log are pairs of the form $E=(d,\theta)$, where $d$ is an operation and $\theta$ its timestamp.
    
    Upon calling an operation $c$, the thread first reads the global clock and stores its value in variable $t$.
    Then,
    \textbf{(1)} If $c$ is a left-mover, the thread appends an entry $E=(c,t)$ to its log, applies $c$ to the local state and stores the response value in variable $v$.
    Otherwise, \textbf{(2)} if $c$ is not a left-mover, the threads append an entry $E=(c,\bot)$ to its log.
    The thread then \textbf{(2a)} fetches and increments the global clock, and executes a compare-and-swap to change the timestamp of $E$ to the retrieved time value.
    Further, \textbf{(2b)} the thread fetches the timestamp assigned to $c$ (by itself or another thread) and stores it in variable $t$.
    The thread \textbf{(2c)} scans the logs, looking for pending operations, i.e., operation whose timestamp is still $\bot$.
    For each such operation, it executes step 2a above to assign it a timestamp.
    This helping mechanism is repeated for each other thread in the order of its log until a timestamp higher than $t$ is assigned.
    Then \textbf{(2d)}, the process finds in the logs all the operations that are between its local clock and $t$.
    Operations which have not been applied so far are applied to the local copy in the order of their timestamp.
    When doing so, the response value of operation $c$ is stored into variable $v$.
    In a last step \textbf{(3)}, the local clock is changed to $t$ before the content of variable $v$ is returned as the result of the operation.
  \end{construction}

  The above construction is wait-free because, when helping in step 2c, a process does so until it assigns a timestamp higher than variable $t$.  
  In what follows, we prove that it is also linearizable.

  For that purpose, consider a run $\lambda$ of the above construction and the corresponding history $h$.
  Since the construction is wait-free, we may always extend $\lambda$ to complete any pending operation.
  Hence, without lack of generality, we shall assume that $h$ is complete.
  In the run $\lambda$, a left-mover $c$ is \emph{seen} by an operation $d$ when $c$ is applied at the time the response of $d$ is computed, i.e., when the calling process executes step 2d.

  Let $L$ be the sequence of operations that are not left-movers in $h$ in the order of their timestamps.
  From $L$, we create a sequence $L'$ as follows:
  For each operation $c \in L$, we append to $L'$ all the left-movers in $h$ seen by $c$ and not yet in $L'$.
  Then we append operation $c$ to $L'$.
  This process is repeated for all the operations in $L$ in the order of the sequence ($<_L$).
  From $L'$, we then obtain $\hat{L}$ by re-ordering the left-movers in $L'$ according to real-time.
  In detail, we first set $\hat{L}$ to $L'$.
  Then, for any two left-movers $c$ and $d$ in $\hat{L}$, if $c \hb_{h} d$ but $d <_{\hat{L}} c$, operations $c$ and $d$ are swapped in $\hat{L}$.
  Let $l$ be the sequential history induced by $\hat{L}$.
  Below, we show that $l$ is a linearization of $h$.

  In the above construction, the local state is always uniquely defined by the sequence of operations that was applied so far.
  We can formulate the following two invariants about the construction:
  \begin{inparaenumorig}
  \item[\emph{(INV1)}] if $c$ happens before $d$ and $d$ has a timestamp, then $d$ has a timestamp and this timestamp is lower than $c$'s; and 
  \item[\emph{(INV2)}] if $c$ has a lower timestamp than $d$ and $d$ is not a left-mover, then $d$ sees $c$.
  \end{inparaenumorig}
  Proving these invariants is straightforward and thus omitted.

  First of all, consider two operations $c \hb_{h} d$.  
  Because of INV1, $c$ precedes $d$ in $L$, leading to $c \hb_l d$, as required.

  Second, we show that the return value of an operation $c$ in $l$ is the same as in $h$.
  This is straightforward in case $c$ is a left-mover.  
  Indeed, whatever is the local state, the operation returns always the same response, the one from the initial state (see~\refsection{dist:scalability}).
  Hereafter, assume that $c$ is not a left-mover.
  
  Let $p$ be the thread executing $c$ in $h$.
  If $c$ sees $d$ then by construction $d$ is before $c$ in $L'$.
  Thus, this is also the case in $\hat{L}$ unless $d$ was swapped with an operation $e$ that is after $d$.
  If this happens then $e \hb_h d$.
  However, in such a case, operation $e$ is seen by $c$ by INV2.
  Hence, it must precedes $c$ in $L'$;
  contradiction.
  
  Next, we prove that every operation before $c$ in $\hat{L}$ is indeed seen by $c$.
  Consider some operation $d$ preceding $c$ in $\hat{L}$.
  There are two cases to consider,
  \textit{(i)} $d$ precedes $c$ in $L'$,
  and \textit{(ii)} $c$ precedes $d$ in $L'$.
  If (i) holds, then by construction $c$ sees $d$, as required.
  Otherwise (ii), operation $d$ is swapped with an operation $e <_{L'} c$ because $d \hb_{h} e$.
  By INV2, because $c$ sees $e$, it must also see $d$.
  This contradicts that $d$ is after $c$ in $L'$.
  
  From what precedes, we know that the operations before $c$ in $\hat{L}$ are the ones seen by $c$.
  We use this to prove that the return value in $l$ is the one in $h$.
  Let $x$ be the sequence of operations applied before $c$ to compute its response value (step 2d).
  Let $y$ equal $\hat{L}_{|<c}$, that is all the operations before $c$ in $\hat{L}$.
  We have established that $x$ and $y$ contains the exact same operations.
  By INV2, operations that are not left-movers are in the same order in $x$ and $y$.
  Thus the two sequences may only differ wrt. left-movers.
  Let $d$ be the first left-mover in $x$ misplaced wrt. $y$.
  We move forward $d$ in $x$ by successively swapping it with the operation right before it until it reaches its position in $y$.
  Let $\hat{x}$ be the resulting sequence of operations.
  Because $d$ is a left-mover, $\tau(s_0,x).\st = \tau(s_0,\hat{x}).\st$ (see~\refsection{dist:scalability}).
  Thus, $\tau(s_0,x.c).\val = \tau(s_0,\hat{x}.c).\val$.
  By applying inductively the above reasoning, we conclude that the response value of $c$ in $l$ and $h$ is the same.
\end{proof}

\invisible*

\begin{proof}
  To implement object $O$, threads use an unbounded, wait-free queue $arr$.
  It permits to append a value ($\offer$), retrieve the index of its last element ($\last$), and access the element at a specific index ($\get$).
  In addition to $arr$, each thread maintains a local copy of the object and a pointer $pnt$, which indicates the index of the last operation in $arr$ applied to the local copy.
  An operation $c$ executes as follows:

  \begin{construction}
    If $c$ is not a right-mover, the thread announces $c$ by appending it to $arr$.
    Then, it applies all the operations (not yet applied) before $c$ in the queue to its local copy of the object.
    Otherwise, the thread does not announce the operation $c$.
    Instead, it locally applies all the operations in $arr$ that are not-yet-applied, stopping at the last operation observed at the time of the invocation.
    In all cases, the thread then locally applies operation $c$ and returns its result as the response to the invocation.
    During this computation, the pointer $pnt$ is updated appropriately.
  \end{construction}
    
  The above construction ensures that right-movers remain invisible.
  Since the objects used in this implementation are wait-free and threads perform a finite number of steps per invocation, every call to object $O$ is wait-free.

  Consider an execution $\run$ and its induced history $h$.
  (As previously, we shall assume that $h$ is complete.)
  We linearize the operations based on their accesses to the shared variable $arr$.
  In detail, the linearization point of the non-right-mover operations in $h$ occurs when they are added to $arr$ (using $\offer$), while for right-movers, it occurs when the last element in $arr$ is observed (using $\last$).
  Let $l$ denote such a linearization.
  By the choice of the linearization points, $\hb_h \subseteq \hb_l$.
  
  Next, we demonstrate that the histories $l$ and $h$ are equivalent.
  Let $c$ be an operation.
  We notice that all operations applied at the time of computing $c$'s response in $\run$ precede $c$ in $l$.
  We successively move all operations that precede $c$ in $l$ but were not applied when computing its response in $\run$ after $c$.
  Necessarily, such operations are right-movers.
  By definition, moving a right-mover $d$ to the right of an operation $e$ does not change the result of $e$, nor the state after $de$.
  Let $l_c$ be the sequential history obtained from this process.
  The result of $c$ in $l_c$ is the same as in $l$.
  Furthermore, by construction, it is identical to $c$'s result in $h$.
  From the above, we can conclude that the construction is linearizable.  
\end{proof}

\subsection{Adjusted objects}
\labappendix{proofs:adjusted}

\computability*

\begin{proof}
  Let $\Delta$ be a distributed task.  
  Consider an algorithm $\mathcal{A}$ that solves $\Delta$ using instances of $O$ and registers.  
  We obtain the algorithm $\mathcal{A}'$ by replacing each instance of $O$ with an instance of $O'$.  
  The algorithm $\mathcal{A}'$ is valid because, as $O$ adjusts $O'$:  
  \emph{(i)} every method that exists in $O$ also exists in $O'$ because $O'.T$ is a (restricted) subtype of $O.T$, and  
  \emph{(ii)} for any thread $p$, if $p$ can call the method $f$ of $O$, then it can also do so with $O'$ (because $O.m \subseteq O'.m$).  
  Now consider an execution $\run$ of $\mathcal{A}$.  
  Let $h$ be the history associated with $\run$.  
  Denote by $l$ a linearization of $h$.  
  By Liskov's substitution principle, since $O'.T$ is a subtype of $O.T$, for every instance $x$ of $O'.T$, the history $l|x$ remains a valid sequential history for $O.T$.  
  (That is, the operations return the same result values.)  
  Thus, the history $h$ is also a history of the algorithm $\mathcal{A}$.  
  Consequently, $\Delta$ is solved in $\run$.
\end{proof}

\scalability*

\begin{proof}
  Clearly, all the vertices in $G_{O.T}(B,s)$ are also in $G_{O'.T}(B,s)$.
  Consider some permutation $x$ of $B$ and let $h$ be the sequential run obtained by applying $x$ from $s$ according to $O.T$.
  Due to Liskov substitution principle, $h$ is also the sequential run obtained by applying $x$ from $s$ according to $O'.T$.
  Hence, for any $(x,y)$ linking two permutations $x$ and $y$ in $G_{O.T}(B,s)$, $(x,y)$ is also in $G_{O'.T}(B,s)$,
\end{proof}

%% file: figures-power.tex
\tikzstyle{run}=[->, decorate, decoration={snake, amplitude=.4mm}]

\begin{figure*}[t]
  \footnotesize
  \begin{tikzpicture}[scale=1.2,auto]
    % Define the nodes
    \node (init) at (2,2) {$\mathcal{C}_0$};
    \node (b) at (2,1) {$\mathcal{C}$};
    \node (pc) at (1,0) {$\mathcal{C}'$};
    \node (qc) at (3,0) {$\mathcal{C}''$};

    \node (top) at (4,2.5) {};
    \node (bot) at (4,-3) {};

    \node (init2) at (8,2) {$\mathcal{C}_0$};
    \node (s) at (8,1) {$\mathcal{C}$};
    
    \node (sg) at (6,0) {$\mathcal{C}$};
    \node (spg) at (5,-1) {$\mathcal{C}'$};
    \node (n1g) at (7,-1) {$q$ decides};
    \node (n2g) at (6,-2) {$q$ decides};
    
    \node (sd) at (10,0) {$\mathcal{C}$};
    \node (spd) at (9,-1) {$\mathcal{C}'$};
    \node (sppd) at (11,-1) {$\mathcal{C}''$};
    \node (n2d) at (10,-2) {$\mathcal{C}''$};
    
    % Draw the curves
    \draw[run] (init) -- (b) node[midway,right] {};
    \draw[run] (init2) -- (s) node[midway,right] {};
    
    % Draw the arrows
    \draw[->] (b) -- (pc) node[midway, above left] {\begin{tabular}{r}\grayed{($p$ calls $c$)}\\$\step_p$\end{tabular}};
    \draw[->] (b) -- (qc) node[midway, above right] {\begin{tabular}{l}\grayed{($q$ calls $c$)}\\$\step_q$\end{tabular}};
    
    \draw[->] (sg) -- (spg) node[xshift=.7em, yshift=-.7em, midway, above left] {\begin{tabular}{r}\grayed{($p$ reads $X_p$)}\\$\step_p$\end{tabular}};
    \draw[run] (sg) -- (n1g) node[xshift=-.7em, yshift=-.7em, midway, above right] {\begin{tabular}{l}\grayed{($q$ solo)}\\$\lambda$\end{tabular}};
    \draw[run] (spg) -- (n2g) node[midway, above right] {$\lambda$};
    
    \draw[->] (sd) -- (spd) node[xshift=.7em, yshift=-.7em, midway, above left] {\begin{tabular}{r}\grayed{($p$ updates $X_p$)}\\$\step_p$\end{tabular}};
    \draw[->] (sd) -- (sppd) node[xshift=-.7em, yshift=-.7em, midway, above right] {\begin{tabular}{l}\grayed{($q$ updates $X_q$)}\\$\step_q$\end{tabular}};
    \draw[->] (spd) -- (n2d) node[midway, below left] {$\step_q$};
    \draw[->] (sppd) -- (n2d) node[midway, below right] {$\step_p$};

    \draw[-, dashed] (s) -- (sg) node[midway, above left] {(i)};
    \draw[-, dashed] (s) -- (sd) node[midway, above right] {(ii)};
    
    \draw[-,color=Gray!50] (top) -- (bot);
  \end{tikzpicture}
  \caption{
    \labfigure{power}
    Illustration of \reflem{proofs:power}.
  }
\end{figure*}